\documentclass[reqno]{amsart}

\usepackage{amsmath,amsthm,amssymb}
\usepackage{stmaryrd}          
\usepackage{mathtools}
\usepackage{hyperref}
\hypersetup{colorlinks=true,linkcolor=black,citecolor=black,urlcolor=black}
\usepackage{enumitem}
\usepackage{microtype}
\usepackage{proof}

\newtheorem{theorem}{Theorem}[section]
\newtheorem{lemma}[theorem]{Lemma}
\newtheorem{corollary}[theorem]{Corollary}
\newtheorem{proposition}[theorem]{Proposition}
\theoremstyle{definition}
\newtheorem{definition}[theorem]{Definition}
\newtheorem{example}[theorem]{Example}
\theoremstyle{remark}
\newtheorem{remark}[theorem]{Remark}

\newcommand{\base}[1]{\mathcal{#1}}
\newcommand{\Bas}{\base{B}}
\newcommand{\BasC}{\mathcal{C}}

\newcommand\supp{\Vdash}
\newcommand{\suppB}[1]{\supp_{\!#1}}

\newcommand{\baseext}{\sqsupseteq}

\newcommand{\enc}[1]{\llbracket#1\rrbracket}

\title{Support is Search}

\author{Alexander V.\ Gheorghiu}
\address{University College London}
\email{a.v.gheorghiu@ucl.ac.uk}


\begin{document}

\begin{abstract}
For intuitionistic propositional logic, Sandqvist's base-extension semantics
defines a support relation parametrised by atomic bases, with validity identified
as support in every base. Sandqvist's completeness theorem answers the
\emph{global} question: which formulae are valid? This paper addresses the
\emph{local} question: given a fixed base, what does support in that base
correspond to? We show that support in a fixed base coincides with proof-search
in a hereditary Harrop logic program with a quantifier over atoms, via an encoding of formulae
as logic programming goals. The encoding is a variant of the Russell--Prawitz translation of the connectives, with
quantification restricted to atoms, and it is obtained by reading the semantic
clauses in continuation-passing style. So read, the universal quantifiers over
base extensions and atoms in those clauses need not range over a completed
totality: they are eigenvariables under a standard freshness discipline. Base-extension semantics thereby acquires an
explicit computational content: support is proof-search. The result
complements Sandqvist's global theorem with a local correspondence and makes the
constructive character of the framework explicit.
\end{abstract}

\maketitle

\section{Introduction}
\label{sec:intro}

The semantics of the logical constants admits, broadly, two styles of
explanation. On the first, model-theoretic semantics, meaning is given
by the conditions under which a sentence is ``true'' in a structure. On
the second, proof-theoretic semantics, meaning is given by the
inferential behaviour that constitutes a demonstration of it; the
tradition is that of Prawitz~\cite{Prawitz1971ideas,Prawitz1973towards},
Dummett~\cite{Dummett1976,Dummett1978,Dummett1991logical}, and
Schroeder-Heister~\cite{Schroeder2006validity,Schroeder2007modelvsproof,SEP-PtS}.
The latter has received less attention, and its computational character in
particular remains largely unexamined. This paper examines that character in a central case.

The case is the \emph{base-extension semantics} for
intuitionistic propositional logic due to
Sandqvist~\cite{Sandqvist2015IL}, a particularly successful realisation
of proof-theoretic semantics. A \emph{base} $\base{B}$ is a
collection of rules governing atomic sentences from which one defines
a derivability relation on atoms $\vdash_{\base{B}}$. From this one has an inductively defined
\emph{support} relation $\supp_{\base{B}}$ on formulae, the clauses of
which define each connective. The clause for falsum,
for example, is
\[
\supp_{\base{B}} \bot
\quad\text{iff}\quad
\supp_{\base{B}} A \ \text{for every atom } A.
\]

The clause is rather demanding computationally. The atoms may be
infinitely many, and a verification of $\supp_{\base{B}} \bot$
accordingly comprises infinitely many verifications --- and this within a
programme whose founding commitment is that the meaning of a sentence is exhausted by what one can do with
it. The constructivist has a standard reply: a universal claim is a
claim to a method, and to say that every atom is supported is to say
that one has, for any given atom, a way of showing that it is. What is that method? The clause does not say. That is the question of this paper. 

The clause for disjunction is more remarkable still. It is \emph{not}
the expected
\[
\supp_{\base{B}} \varphi \lor \psi \quad \text{iff} \quad
\supp_{\base{B}} \varphi \text{ or } \supp_{\base{B}} \psi
\]
--- which would yield not intuitionistic but inquisitive logic, which is strictly stronger (see
Piecha et
al.~\cite{Piecha2015failure,Piecha2016completeness,Piecha2019incompleteness};
Stafford~\cite{Stafford2021}) --- but rather the `second-order
condition'
\[
\begin{array}{ll}
\supp_{\base{B}} \varphi \lor \psi \quad \text{iff} &
\text{for every } \base{C} \baseext \base{B}\ \text{and every atom } A, \\[2pt]
& \text{if } \varphi \supp_{\base{C}} A\ \text{and}\ \psi \supp_{\base{C}} A,
  \ \text{then } \supp_{\base{C}} A,
\end{array}
\]
where $\base{C} \baseext \base{B}$ means that $\base{C}$ extends
$\base{B}$ by further rules, and $\varphi \supp_{\base{C}} A$ (resp.  $\psi \supp_{\base{C}} A$) means
that $A$ is supported in every extension of $\base{C}$ in which
$\varphi$ (resp. $\psi$) is. Whether $A_1 \lor A_2$ is supported in $\base{B}$ --- where $A_1$ and $A_2$ are atoms --- thus
depends upon every atom of the language, including those bearing no
relation whatever to $A_1$ or $A_2$, and upon every extension of
$\base{B}$. Whether that demand is met seems, \emph{prima facie}, rather
difficult to verify. 

The two clauses raise a single question. 
Sandqvist's theorem is that $\supp_{\base{B}} \varphi$ for every base
$\base{B}$ if and only if $\varphi$ is intuitionistically provable. The
result is striking, but it is also \emph{global}: it characterises the
formulae supported in every base. It leaves untouched what is, from a
foundational point of view, the prior question --- \emph{what does
$\supp_{\base{B}} \varphi$ amount to in a fixed base $\base{B}$?}

The naive answer has
much to recommend it. At atoms, support simply is derivability:
$\supp_{\base{B}} A$ holds precisely when $A$ is derivable by the rules
of $\base{B}$. Support is monotone in the base, as derivability is in
its hypotheses; the clause for implication, $\supp_{\base{B}} \varphi
\to \psi$ iff $\varphi \supp_{\base{B}} \psi$, has the form of a
deduction theorem; and a base, being a set of rules, may be read as a
set of hypotheses. One expects, accordingly, that $\supp_{\base{B}}
\varphi$ should coincide with derivability of $\varphi$ from some
natural translation of $\base{B}$ into a set of formulae, and that the
universal quantifiers of the clauses are merely a circuitous way of
saying as much.

The expectation is refuted by an example that Sandqvist
gives~\cite{Sandqvist2015IL}. Fix atoms $A_1$ and $A_2$, and let $\base{R}$
consist of the rules
\[
\infer{A}{\deduce{A}{[A_1]} & \deduce{A}{[A_2]}}
\]
one for each atom $A$: from $A$-if-$A_1$ and $A$-if-$A_2$, infer $A$. The
disjunction $A_1 \lor A_2$ is supported in $\base{R}$: in any extension
in which $A$ follows from $A_1$ and from $A_2$, the rule for $A$ concludes
it outright. Yet neither $\supp_{\base{R}} A_1$
nor $\supp_{\base{R}} A_2$ holds: a derivation of $A_1$ in $\base{R}$ could
end only with the rule for $A_1$, one of whose premisses is a derivation
of $A_1$ from $A_2$, which could in turn end only with that same rule, and
the same premiss remains outstanding; the case of $A_2$ is symmetric.

Derivability cannot behave in this way. The rules of $\base{R}$, read
as formulae, are $\lor$-free --- Harrop formulae --- and from Harrop
hypotheses a disjunction is derivable only if one of its disjuncts
is~\cite{harrop1960concerning}. What fails is the principle Piecha and
Schroeder-Heister call \emph{export}~\cite{Piecha2019incompleteness}:
that the rules of a base may be exported into hypotheses, so that
$\supp_{\base{B} \cup \Sigma} \varphi$ just in case $\Sigma^{*}
\supp_{\base{B}} \varphi$, where $\Sigma^{*}$ is the set of rules of
$\Sigma$ read as formulae, and hence, taking $\base{B}$ empty, that
$\supp_{\Sigma} \varphi$ just in case $\Sigma^{*} \vdash \varphi$.
The lesson is that support in a base is not simply derivability in disguise.

We propose an alternative reading. The clause for disjunction does not
say when a disjunction is supported; it says what one may \emph{do}
with one --- namely, conclude any atom that each disjunct would yield
on its own. It has, in this respect, a recognisable form in functional
programming as a \emph{continuation-passing} definition: the atom $A$
is the answer sought, and the hypotheses $\varphi \supp_{\base{C}} A$
and $\psi \supp_{\base{C}} A$ are the continuations --- the branches of
a case analysis, one of which will deliver it. A type theorist will recognise in the clause the type
\[
\forall \rho\,((\alpha \to \rho) \to (\beta \to \rho) \to \rho)
\]
 that a
sum acquires under continuation-passing conversion. A logician will recognise it as
Russell's definition of disjunction~\cite[\S 19]{russell1903} and
Prawitz's second-order encoding of the
connectives~\cite{Prawitz2006natural}. The recognition matters for what it suggests: that the clauses may be
read computationally after all, and that the place to read them is a
paradigm at once logical and computational --- \emph{logic
programming}.

Logic programming, in the tradition of Kowalski~\cite{kowalski1979logic}
and, in the form relevant here, of Miller et
al.~\cite{miller1991uniform,miller2012programming}, is the view of
computation as proof-search. A \emph{program} is a set of formulae of
restricted form --- \emph{program clauses} --- and a \emph{goal} is a formula
whose provability from the program is in question; to \emph{execute}
the goal is to search for a proof of it, and the search is driven by
the logical structure of the goal, each connective being an instruction
as to how to proceed. A conjunction instructs one to establish both
conjuncts; a disjunction, to establish either; an implication $D
\supset G$, to add $D$ to the program and establish $G$; an atom, to
consult the program. The program is thus consulted only at atoms, and
the connectives of the goal are operations of search.

The bearing on
the present problem is immediate: a base is a set of rules for atoms
--- that is, a program --- and the support clauses, read as above, are
themselves goals: instructions as to how a formula is to be verified
against it. The instructions are not always those of the like-named
connectives. Support for a disjunction does not require support for
either disjunct, and the clause for disjunction is accordingly not a
disjunctive goal; it is a universally quantified implication, which
instructs one to take an arbitrary atom, to assume that it follows from
each disjunct, and to establish it.

In
logic programming, moreover, a universal quantifier in a goal is not a
claim about a totality at all but an instruction: introduce a fresh
name and proceed~\cite{miller1991uniform}. If the quantifiers of
Sandqvist's clauses admit this reading, the infinities disappear, and
each support clause says precisely what a verification consists in --- which is
what was asked for.

To each formula $\varphi$ we associate a goal $\enc{\varphi}$
of the propositional hereditary Harrop language of Miller et
al.~\cite{miller1991uniform}, extended by a quantifier over atoms, and
we establish that, for every base $\base{B}$,
\[
\supp_{\base{B}} \varphi
\quad\text{iff}\quad
\base{B} \vdash_O \enc{\varphi},
\]
where $\vdash_O$ is the operational semantics of the language --- the
relation that obtains when $\enc{\varphi}$, executed as a goal against
the program $\base{B}$, succeeds (Theorem~\ref{thm:snc}). The encoding
is a variant of the Russell--Prawitz translation, and the proof pairs each universal
quantifier of the semantic clauses with a step of proof-search: the
quantifiers over base extensions become \textsc{Augment}, which adds a
hypothesis to the program, and the quantifiers over atoms become
\textsc{Generic}, which introduces an atom under the usual freshness
condition and proceeds. The constructive reading promised a method;
the theorem exhibits it. Support is search.

The correspondence is not to be had for nothing, and three points
require care. The encoding of an implication, $\enc{\varphi} \supset
\enc{\psi}$, must serve both as a goal and as a clause of a program,
and its consequent is in general compound; the clausal language must
accordingly admit program clauses with compound heads, as the hereditary Harrop
formulae of Miller et al.~\cite{miller1991uniform} do, such clauses
being resolved into clauses with atomic heads when the program is
consulted (Definition~\ref{def:op-prov}).

The quantifier of the
language is instantiated at atoms only. It is the quantifier of
Ferreira and Ferreira's atomic polymorphism~\cite{Ferreria2013}, not
that of second-order logic. So the
substitution lemma one would naturally reach for is false: a universal
goal may be derivable while a compound instance of it is not
(Remark~\ref{rem:nosubst}).

The proof therefore proceeds by renaming
of atoms, together with identity and cut for the class of formulae
that are at once program clauses and goals; it follows the method of Miller's
argument that his Kripke models of hereditary Harrop programs agree
with uniform provability~\cite{miller1992abstract,miller2025ptlp}. What
it adds is the treatment of Sandqvist's clauses for $\lor$ and $\bot$,
which are elimination conditions quantifying over atoms and not the
satisfaction clauses of a Kripke model, and of the atomic quantifier
that answers to them.

The result is a local correspondence theorem
complementing Sandqvist's global one. It has three consequences.

First, the semantic clauses need not be regarded as ranging over
completed infinite collections of inferential contexts. Their content
is exhausted by a finitary search procedure --- a successful search
being a finite derivation. Proof-theoretic semantics, anti-realist in
motivation, calls for a constructive reading. That reading is supplied
not by external argument but by the structure of the semantics itself.

Second, support acquires computational content of a definite kind. The
relation $\supp_{\base{B}} \varphi$ is not merely constructive in some
abstract sense; it is the input--output specification of an explicit
procedure --- execute $\enc{\varphi}$ as a goal against $\base{B}$ ---
whose successful runs are uniform derivations witnessing support. The
semantics is, in this sense, executable; and since, by Sandqvist's
theorem, $\varphi$ is intuitionistically provable just in case
$\enc{\varphi}$ succeeds against the empty program
(Corollary~\ref{cor:validity}), so too is intuitionistic provability.

Third, and most broadly, the result places proof-theoretic semantics
and the proof theory of logic programming in unexpected proximity.
Sandqvist's construction was developed independently of, and later than, Mints' resolution-based proof methods for intuitionistic
logic~\cite{Mints1988}; we have shown that the two
constructions coincide up to renaming~\cite{MINTS}. The present paper extends this
convergence: it is not only Sandqvist's particular base that is
computationally intelligible, but the semantic framework in which it
is constructed. What was offered as a philosophical foundation turns
out to be, in mathematical content, a chapter of computational logic.

\subsection*{Organisation}

Section~\ref{sec:atomic-systems} develops the framework of atomic
systems and hereditary Harrop programs with a quantifier over atoms, and recalls the
correspondence between them established by Gheorghiu and
Pym~\cite{NAF}. Section~\ref{sec:bes} recalls Sandqvist's
base-extension semantics, the relevant background results, and the
local question. Section~\ref{sec:snc} introduces the encoding of
intuitionistic formulae as logic programming goals and proves the main
theorem. Section~\ref{sec:discussion} discusses the implications of
the result --- in particular for the question of what the
universal quantifiers in base-extension semantics mean, for the role of
higher-level rules, and for the prospects of a uniformly computational
reading of proof-theoretic semantics more broadly.

\section{Atomic Systems and Clausal Programs} \label{sec:atomic-systems}

Base-extension semantics rests on a primitive notion of atomic
inference: before any account is given of the logical constants, one
must first say what it is for an atomic sentence to be derivable.
This is the role of an \emph{atomic system}, or \emph{base}. The
purpose of this section is to fix the framework of atomic systems we
shall use, and to develop a computational presentation of it. Two
moves are involved.  First, following Prawitz~\cite{Prawitz1971ideas,Prawitz1973towards}
and Piecha and Schroeder-Heister~\cite{Schroeder2016atomic,Piecha2017definitional},
we admit atomic rules of arbitrary finite level --- rules whose
premisses may themselves be rules. Second, following Halln\"as and
Schroeder-Heister~\cite{hallnas1990proof,hallnas1991proof} and Miller
et al.~\cite{miller1989logical,miller1991uniform}, we re-present these
rules as clauses of a logic programming language, with derivability in
the atomic system corresponding to operational proof-search in the
clausal program.

The first move lets the semantics characterise intuitionistic rather
than classical logic: Sandqvist has
shown that first-level rules alone yield classical
logic~\cite{Sandqvist2009CL}, while rules of level two suffice for
the intuitionistic case~\cite{Sandqvist2015IL}. We admit rules of
arbitrary finite level for uniformity with the clausal language. The
second move makes the main theorem of the paper possible. Both are well established in the
literature, but they are not customarily presented together, and the
combination repays explicit treatment.

Fix an infinite set $\mathrm{At}$ of \emph{atoms}, or atomic sentences, ranged
over by $A, B, C, \ldots$. The infinity at issue is \emph{potential},
not actual: at any stage of a derivation only finitely many atoms have
been introduced, but no bound is placed on how many may be introduced
later. This is the only notion of infinity admissible from an
anti-realist standpoint, and the only one we shall use --- it
presupposes no completed totality, only the open-ended availability of
fresh names as derivation proceeds.

\subsection*{Higher-level atomic systems}

The simplest atomic systems consist of \emph{production rules}: rules
that license the inference of an atom from a finite collection of
atoms.

\begin{definition}[First-level atomic system]
A \emph{first-level atomic system} $\base{S}$ is a set of rules of the form
\[
  \frac{A_1 \quad \cdots \quad A_n}{C}
\]
in which $A_1, \ldots, A_n$ and $C$ are atoms.
Derivability in $\base{S}$ --- written $\vdash_{\base{S}}$ --- is the
least relation closed under the rules.
\end{definition}

This format is, however, too restrictive. Rules in natural deduction
for the logical constants --- implication introduction and disjunction
elimination, in particular --- are not of this shape: their premisses
are not assertions but \emph{derivations from hypotheses}. If the
atomic level is to support a semantics flexible enough to deliver
intuitionistic logic, it must admit rules of the same kind. The
generalisation, due to Piecha and
Schroeder-Heister~\cite{Schroeder2016atomic,Piecha2017definitional},
allows the premisses of an atomic rule to be themselves finite sets of
atomic rules; iterating yields rules of arbitrary finite level.

\begin{definition}[Higher-level atomic system]
An \emph{atomic rule of level $0$} is an axiom 
\[
\frac{}{~~B~~},
\]
where $B$
is an atom. An \emph{atomic rule of level $1$} is a production rule
\[
  \frac{A_1 \quad \cdots \quad A_n}{B}
\]
on atoms. For $n \geq 1$, an \emph{atomic rule of level $n+1$} has the
form
\[
  \infer{B}{\deduce{A_1}{[\Sigma_1]} & \ldots & \deduce{A_k}{[\Sigma_k]}}
\]
where $B$ and each $A_i$ are atoms and each $\Sigma_i$ is a finite set
of atomic rules of level at most $n-1$. Thus a rule of level $2$
discharges atoms, a rule of level $3$ discharges production rules, and
so on. The rule is read as licensing
the conclusion $B$ whenever each $A_i$ is derivable from the rules in
$\Sigma_i$, with those rules thereby discharged. A \emph{higher-level
atomic system} is a set of atomic rules of level at most $n$ for some
fixed $n$, and $\vdash_{\base{S}}$ is again the least relation closed
under the rules.
\end{definition}

\begin{example}[A third-level rule]
\label{ex:third-level}
Let $\base{S}$ contain the level-$3$ rule
\[
  \infer{D}{\deduce{C}{\left[\raisebox{-1ex}{\infer{B}{A}}\right]}}
\]
together with the level-$2$ rule
\[
\infer{C}{\deduce{B}{[A]}}
\]
The first rule licenses $D$ whenever $C$ is derivable using the
production rule from $A$ to $B$, that rule being discharged. The second rule licenses $C$ whenever $B$ is derivable from the
assumption $A$. With the production rule from $A$ to $B$ in hand, it
is; so $C$ is derivable using that rule, and the first rule yields $D$
in $\base{S}$.
\end{example}

\subsection*{Rules as clauses}

The presentation of atomic systems we have just given is faithful to
their natural deduction origins, but for the purposes of this paper a
different presentation is more convenient. We re-present atomic rules
as clauses of a logic programming language, and atomic derivability as
operational proof-search. The translation, which we use throughout, is due in essence to
Halln\"as and
Schroeder-Heister~\cite{hallnas1990proof,hallnas1991proof}.

The clausal language is the propositional fragment of the
\emph{hereditary Harrop} formulae of Miller et
al.~\cite{miller1989logical,miller1991uniform}.

\begin{definition}[Definite formulae, goals, programs]
The language is given by mutual recursion. \emph{Definite formulae}
$D$, which appear in programs, and \emph{goal formulae} $G$, which are
queried against them, are generated by the grammar
\[
  D \;::=\; A \;\mid\; G \supset D \;\mid\; D \cap D
  \qquad\qquad
  G \;::=\; A \;\mid\; D \supset G \;\mid\; G \cap G \;\mid\; G \cup G,
\]
where $A$ ranges over atoms. A \emph{program} $P$ is a finite set of
definite formulae; a \emph{query} is a pair of a program $P$ and a goal $G$.
We adopt the convention that $\cap$ and $\cup$ bind more tightly than
$\supset$, so that $G_1 \cap G_2 \supset D$ abbreviates $(G_1 \cap G_2)
\supset D$. In a definite formula $G \supset D$, the formula $D$ is the
\emph{head}.
\end{definition}

\begin{remark}
  We use the signs $\supset$, $\cap$, $\cup$ for the connectives of the
clausal language in order to reserve the more standard $\to$, $\land$,
$\lor$ for the object logic. They are to be read intuitively as
\emph{operational} implication, conjunction, and disjunction, governed
by the proof-search rules of Definition~\ref{def:op-prov}. The
emphasis on `operational' as opposed to `logical' is crucial: it is
logic itself that is here under scrutiny, and were the clausal
connectives understood logically, the explanation would move in a
circle. The reader who finds this distinction elusive is referred to
Dummett~\cite{Dummett1978} for its philosophical articulation. In this sense the
clausal language is \emph{pre-logical} --- a calculus of inference
about atoms, on which a semantics for the logical constants will
subsequently be erected.
\end{remark}

The asymmetry between $D$ and $G$ is the source of the language's good
proof-theoretic behaviour: definite formulae allow conjunction but not
disjunction; goal formulae allow both, but their implications are
\emph{hypothetical} --- the antecedent is a definite formula, added to
the program for the duration of the goal.

The intended reading of a query is procedural:
to resolve $G$ against $P$, decompose $G$ according to its logical
structure, consulting $P$ only when an atomic goal is reached. This
is captured by the following operational semantics.

\begin{figure}[t]
\hrule
\vspace{2mm}
\centering
\renewcommand{\arraystretch}{1.45}
\begin{tabular}{@{}l@{\quad}l@{\quad}l@{}}
  \textsc{Fact}      & $P \vdash_O A$             & if $A \in [P]$ \\
  \textsc{Backchain} & $P \vdash_O A$             & if $G \supset A \in [P]$ and $P \vdash_O G$, for some $G$ \\
  \textsc{And}       & $P \vdash_O G_1 \cap G_2$  & if $P \vdash_O G_1$ and $P \vdash_O G_2$ \\
  \textsc{Or}        & $P \vdash_O G_1 \cup G_2$  & if $P \vdash_O G_1$ or $P \vdash_O G_2$ \\
  \textsc{Augment}   & $P \vdash_O D \supset G$   & if $P \cup \{D\} \vdash_O G$ \\
  \textsc{Generic}   & $P \vdash_O \forall X\, G$ & if $P \vdash_O G[A/X]$, for some $A$ fresh for $P$ and $G$
\end{tabular}
\vspace{2mm}
\hrule
\caption{Rules of search. The rule \textsc{Generic} belongs to the
extension by the atomic quantifier (Definition~\ref{def:atomic-quant}).}
\label{fig:op-sem}
\end{figure}

\begin{definition}[Operational semantics]
\label{def:op-prov}
The relation $\vdash_O$ between programs and goals is defined by the
rules \textsc{Fact}, \textsc{Backchain}, \textsc{And}, \textsc{Or}, and
\textsc{Augment} of Figure~\ref{fig:op-sem}. Here $[P]$, the \emph{elaboration} of $P$, is the least set of definite
formulae that contains $P$ and satisfies the following conditions:
\begin{itemize}
  \item if $D_1 \cap D_2 \in [P]$, then $D_1, D_2 \in [P]$;
  \item if $G \supset (D_1 \cap D_2) \in [P]$, then $G \supset D_1,\, G
        \supset D_2 \in [P]$;
  \item if $G \supset (G' \supset D) \in [P]$, then $G \cap G' \supset D
        \in [P]$.
\end{itemize}
The formulae of $[P]$ of the forms $A$ and $G \supset A$ are the
\emph{resolved clauses} of $P$.
\end{definition}
\begin{remark}
Though Definition~\ref{def:op-prov} is stated deductively, the relation is best thought of as \emph{reductive} in the sense of Pym and Ritter~\cite{pym2004reductive}. A query is a putative conclusion $P \vdash_O G$. The rules of the operational semantics mean that we can resolve the query by case analysis on the structure of the goal $G$. Thus the goal drives a proof-search, and the
program is consulted only at atoms, where one either appeals to a fact
or backchains through an implication.

The elaboration $[P]$
discharges the bookkeeping of definite conjunctions and of compound
heads: every clause of the program is resolved into clauses with atomic
heads, $G \supset (G' \supset A)$ doing the work of $G \cap G' \supset
A$, and \textsc{Fact} and \textsc{Backchain} consult only these. Each condition on $[P]$ has a single premiss, so elaboration
distributes over union: $[P \cup Q] = [P] \cup [Q]$. Derivations in
$\vdash_O$ are, in the terminology of Miller et
al.~\cite{miller1991uniform}, \emph{uniform}: the goal is decomposed
before the program is consulted.
\end{remark}

The connection between atomic systems and clausal programs is now
direct: each atomic rule translates to a program clause, and the
level of the rule corresponds to the nesting depth of implications in
that clause.

\begin{example}[Rules as clauses]
\label{ex:rules-as-clauses}
The level-$2$ rule
\[
\infer{C}{\deduce{B}{[A]}}
\]
becomes the program clause $(A \supset B) \supset C$. The reading is the same in both
notations: to derive $C$, it suffices to derive $B$ from the assumption
$A$. The operational semantics confirms this. With $P = \{(A \supset
B) \supset C\}$, \textsc{Backchain} reduces $P \vdash_O C$ to
$P \vdash_O A \supset B$, and \textsc{Augment} reduces this to
$P, A \vdash_O B$:
\[
\infer{P \vdash_O C}{P,A \vdash_O B}
\]
\end{example}

Gheorghiu and Pym~\cite{NAF} have
shown that the hierarchy of higher-level atomic systems corresponds
exactly to the implication-nesting depth of definite formulae: an
atomic rule of level $n$ becomes a definite formula of nesting depth
$n$, and conversely. For any atomic system $\base{S}$ with clausal
representation $P$, and any atom $A$,
\[
  P \vdash_O A \qquad \text{iff} \qquad \vdash_{\base{S}} A.
\]
Clauses with compound heads add convenience, not strength: the
clausal representation of an atomic rule has an atomic head, and a
clause with a compound head is, by the elaboration, a notation for
clauses with atomic heads. Compound heads are admitted because
Section~\ref{sec:snc} requires implications to be at once program clauses and goals whatever their consequents.
We shall henceforth identify an atomic system with its clausal
representation, writing $\base{S} \vdash_O A$ in place of
$\vdash_{\base{S}} A$.

\subsection*{Schematic infinity and the atomic quantifier}

One complication remains. The atomic systems used in base-extension
semantics are routinely \emph{infinite}; indeed, the completeness of
intuitionistic propositional logic with respect to support requires
this. But not every infinite system is to be admitted. This is not a
technical convenience but a substantive constraint: the anti-realist
foundations of the semantics forbid appeal to arbitrary completed
infinities. The infinities of base-extension
semantics must be \emph{schematic}: each instance of a rule must be
specifiable by a finite description.

A representative case illustrates the point. Sandqvist's completeness
proof~\cite{Sandqvist2015IL} uses an atomic system containing
every instance of the schema
\[
  (A \supset X) \cap (B \supset X) \supset X
\]
for some fixed atoms $A$ and $B$ as $X$ ranges over atoms. The system is infinite, but its content is
exhausted by a finite specification. We need only formalise that finite specification. To this end, we extend the logic programming 
language with a quantifier ranging over atoms. Propositional
variables, ranged over by $X$, are drawn from a set $\mathrm{Var}$
disjoint from the atoms.

\begin{definition}[Atomic quantifier]
The grammar of definite formulae and goals is extended by the productions
\[
G\;::=\; \ldots \;\mid\; \forall X\, G
\qquad\qquad
D\;::=\; \ldots \;\mid\; \forall X\, D,
\]
where $X$ is a propositional variable.
\end{definition}

The schema above becomes the single program clause
\[
\forall X\bigl((A \supset X) \cap (B \supset X) \supset X\bigr),
\]
and the move from infinite to finite specification is complete.

The \emph{atomic formulae} of the extended language are the
atoms together with the variables. We write $G[A/X]$ for the result of
replacing the free occurrences of $X$ in $G$ by the atom $A$, and we say
that an atom is \emph{fresh} for a formula or program if it does not occur
in it. Programs, and the goals queried against them, are \emph{closed}:
every variable occurring in them is bound. In a closed formula, therefore,
an atomic subformula standing alone as a goal or as a program clause is an atom.

The quantifier is not the quantifier of second-order propositional logic,
nor the higher-order quantifier of Miller et al.~\cite{miller1991uniform},
either of which may be instantiated at compound formulae of the
appropriate type. It is instantiated at atoms only, and is in this
respect the quantifier of the system $\mathbf{F}_{\mathrm{at}}$ of atomic
polymorphism studied by Ferreira and
Ferreira~\cite{Ferreria2013,ferreira2014faithfulness}; the point is
taken up in Remark~\ref{rem:nosubst}. Its operational behaviour is
governed by the eigenvariable discipline.

\begin{definition}[Operational rules for $\forall$]
\label{def:atomic-quant}
The relation $\vdash_O$ is extended by the rule \textsc{Generic} of
Figure~\ref{fig:op-sem}. The elaboration of programs is extended correspondingly: if $\forall X\, D \in [P]$, then $D[A/X] \in [P]$ for every atom $A$; and if $G \supset \forall X\, D \in [P]$, then $G \supset D[A/X] \in [P]$ for every atom $A$.
\end{definition}

\begin{remark}
Were one
to require that the body of a $\forall$-goal hold for \emph{every}
atomic instantiation, the infinitary problem would simply
recur, displaced from the program to the proof. The freshness
condition replaces this universal quantification with a single
schematic instance --- a witness that, because nothing in the program
or goal mentions it, stands in for any atom whatever. This is
finitary and mechanical, and it is the device that makes the schematic
infinities of the semantics tractable. We shall see
in Section~\ref{sec:snc} that the universal quantifiers in the
support clauses are governed by this same discipline.
\end{remark}

We close the section by recording three properties of $\vdash_O$ which
will be used freely in what follows. Each is proved by induction on the
derivation, and each preserves the height of the derivation, a fact on
which the inductions of Section~\ref{sec:snc} depend. A function
$\sigma$ from atoms to atoms is extended to formulae homomorphically,
leaving variables unchanged, and to programs pointwise; we write
$F\sigma$ and $P\sigma$ for the results. Observe that $[P]\sigma
\subseteq [P\sigma]$: the map $\sigma$ commutes with the conditions
that define the elaboration, and in particular it carries the instance
$D[A/X]$ of $\forall X\, D$ to the instance $(D\sigma)[A\sigma/X]$ of
$\forall X\, (D\sigma)$.

\begin{lemma}[Renaming]
\label{lem:rename}
Let $\sigma$ be any function from atoms to atoms. If $P \vdash_O G$,
then $P\sigma \vdash_O G\sigma$.
\end{lemma}

\begin{proof}
By induction on the height of the derivation of $P \vdash_O G$, the
statement being taken for all $\sigma$ at once. For \textsc{Fact}, $A
\in [P]$ gives $A\sigma \in [P]\sigma \subseteq [P\sigma]$. For
\textsc{Backchain}, $G' \supset A \in [P]$ and $P \vdash_O G'$; the
induction hypothesis gives $P\sigma \vdash_O G'\sigma$, and $G'\sigma
\supset A\sigma \in [P\sigma]$. The cases \textsc{And}, \textsc{Or}, and
\textsc{Augment} are immediate from the induction hypothesis, the last
because $(P \cup \{D\})\sigma = P\sigma \cup \{D\sigma\}$. For
\textsc{Generic}, $G = \forall X\, G'$ and $P \vdash_O G'[A/X]$ for some
atom $A$ fresh for $P$ and $G'$. Choose an atom $B$ fresh for $P\sigma$
and $G'\sigma$, and let $\sigma'$ agree with $\sigma$ except that
$\sigma'(A) = B$. The induction hypothesis applied with $\sigma'$ gives
$P\sigma' \vdash_O G'[A/X]\sigma'$. Since $A$ occurs in neither $P$ nor
$G'$, we have $P\sigma' = P\sigma$ and $G'[A/X]\sigma' =
(G'\sigma)[B/X]$. As $B$ is fresh for $P\sigma$ and $G'\sigma$,
\textsc{Generic} gives $P\sigma \vdash_O \forall X\,(G'\sigma) =
G\sigma$.
\end{proof}

The case of interest is that in which $\sigma$ moves a single atom. If $P
\vdash_O G$ and the atom $A$ does not occur in $P$, then $P \vdash_O
G[B/A]$ for every atom $B$, where $G[B/A]$ denotes the result of
replacing $A$ by $B$ throughout $G$. An eigenvariable introduced by
\textsc{Generic} may thus be moved to any atom whatever, provided the
program does not mention it. This is the only form of substitution the
paper requires; see Remark~\ref{rem:nosubst}.

We write $Q \baseext P$ to mean that $Q \supseteq P$, that is, that $Q$
extends $P$ as a program.

\begin{lemma}[Monotonicity]
\label{lem:mono-op}
If $P \vdash_O G$ and every resolved clause of $P$ is a resolved clause
of $Q$, then $Q \vdash_O G$. In particular, if $P \vdash_O G$ and $Q
\baseext P$, then $Q \vdash_O G$.
\end{lemma}

\begin{proof}
The proof is by induction on the height of the derivation. The cases
\textsc{Fact} and \textsc{Backchain} hold by the hypothesis on resolved
clauses, and the
cases \textsc{And} and \textsc{Or} are immediate. For \textsc{Augment},
elaboration distributes over union, so the resolved clauses of $P \cup
\{D\}$ are those of $P$ together with those of $\{D\}$, and all are
resolved clauses of $Q \cup \{D\}$; the induction hypothesis applies. For
\textsc{Generic}, $P \vdash_O G'[A/X]$ with $A$ fresh for $P$ and $G'$,
but $A$ may occur in $Q$. Choose $B$ fresh for $Q$ and $G'$. By
Lemma~\ref{lem:rename}, $P \vdash_O G'[B/X]$ by a derivation of the same
height, so the induction hypothesis gives $Q \vdash_O G'[B/X]$, and
\textsc{Generic} concludes.
\end{proof}

A conjunction $D_1 \cap D_2$ is not itself a resolved clause, and $[D_1
\cap D_2] = \{D_1 \cap D_2\} \cup [D_1] \cup [D_2]$. The programs $P
\cup \{D_1 \cap D_2\}$ and $P \cup \{D_1, D_2\}$ therefore have the same
resolved clauses, and the lemma gives
\[
P, D_1 \cap D_2 \vdash_O G \qquad \text{iff} \qquad P, D_1, D_2 \vdash_O G.
\]

\begin{lemma}[Inversion]
\label{lem:inv} The following statements hold:
\begin{enumerate}
\item[(i)] If $P \vdash_O D \supset G$, then $P, D \vdash_O G$.
\item[(ii)] If $P \vdash_O G_1 \cap G_2$, then $P \vdash_O G_i$ for $i
= 1, 2$.
\item[(iii)] If $P \vdash_O \forall X\, G$, then $P \vdash_O G[A/X]$ for
every atom $A$.
\end{enumerate}
\end{lemma}

\begin{proof}
The rules \textsc{Fact} and \textsc{Backchain} conclude only atomic
goals, so the last rule of a derivation of a compound goal is determined
by the principal connective of the goal, and (i) and (ii) follow by
reading off its premisses. For (iii), the last rule is \textsc{Generic},
with premiss $P \vdash_O G[A'/X]$ for some $A'$ fresh for $P$ and $G$.
Given any atom $A$, Lemma~\ref{lem:rename} applied to the renaming $A'
\mapsto A$ gives $P \vdash_O G[A'/X][A/A'] = G[A/X]$, since $A'$ occurs
in neither $P$ nor $G$.
\end{proof}

Finally, $\vdash_O$ is closed under cut on atoms: if $P \vdash_O A$ and
$P, A \vdash_O G$, then $P \vdash_O G$. The proof is by induction on the
height of the derivation of $P, A \vdash_O G$. The atom $A$ contributes
only itself to $[P, A]$, so it can be used only by \textsc{Fact}, to
conclude $A$, and there the derivation of $P \vdash_O A$ serves instead;
at \textsc{Augment}, that derivation is first weakened by
Lemma~\ref{lem:mono-op}. Lemma~\ref{lem:cut} extends the property to a
wider class of formulae.

\section{Base-extension Semantics}
\label{sec:bes}

We turn now to the semantics itself. The atomic systems of Section~\ref{sec:atomic-systems} are the ground
on which the logical constants are given meaning. The semantics is due, in the form we present, to
Sandqvist~\cite{Sandqvist2015IL}, and we record his global
completeness theorem before stating the local question that motivates
the present paper.

The language of intuitionistic propositional logic is given by
\[
  \varphi, \psi, \chi \;::=\; A \in \mathrm{At}\;\mid\; \bot \;\mid\;
    \varphi \land \psi \;\mid\; \varphi \lor \psi \;\mid\; \varphi \to \psi,
\]
with $\lnot\varphi \coloneqq \varphi \to \bot$. We write $\Gamma
\vdash \varphi$ for intuitionistic derivability of $\varphi$ from
$\Gamma$.

The semantics is parametrised by atomic systems. In view of the
clausal presentation of Section~\ref{sec:atomic-systems}, we identify
atomic systems with logic programs throughout.

\begin{definition}[Base]
A \emph{base} $\Bas$ is a program in the clausal language
of Section~\ref{sec:atomic-systems}.
\end{definition}

This is not quite Sandqvist's class of bases. His are arbitrary sets of
atomic rules of level at most two~\cite{Sandqvist2015IL}; ours are
finite programs, whose rules may be of any finite level and which admit
infinite families of rules only through the quantifier. The results
quoted from~\cite{Sandqvist2015IL} below ---
Propositions~\ref{prop:mono} and~\ref{prop:cons} and
Theorem~\ref{thm:Sandqvist} --- are stated there for his class. Their
proofs transfer unchanged, because they use only two facts about the
class of bases: that a base extended by atoms or rules is
again a base, and, for completeness, that the flattening base
$\base{N}$ constructed below belongs to the class. The
base $\base{N}$, in particular, consists of finitely many rule schemas,
a few for each subformula of the given formula, and each schema is a
single clause of that language; Gheorghiu and
Pym~\cite{NAF} carry out the construction for clausal programs.

The meaning of the logical constants is given by an inductively
defined relation between bases and formulae --- the \emph{support}
relation --- whose clauses we now state.

\begin{definition}[Support]
The relation $\Gamma \suppB{\Bas} \varphi$ is defined by the clauses
of Figure~\ref{fig:support}, in which $A$ ranges over $\mathrm{At}$ and $\Gamma$ is a finite set of formulae; we write $\varphi \suppB{\Bas} \psi$ for $\{\varphi\} \suppB{\Bas} \psi$.
\end{definition}

\begin{figure}[t]
\hrule
\vspace{2mm}
\begin{align*}
    & \supp_{\base{B}} A  && \text{iff} \quad \base{B} \vdash_O A \tag{At}  \\[1mm]
    & \supp_{\base{B}} \varphi \to \psi && \text{iff} \quad \varphi \supp_{\base{B}} \psi \tag{$\to$} \\[1mm]
    & \supp_{\base{B}} \varphi \land \psi && \text{iff} \quad \supp_{\base{B}} \varphi \text{ and } \supp_{\base{B}} \psi \tag{$\land$} \\[1mm]
    & \supp_{\base{B}} \varphi \lor \psi && \text{iff} \quad \text{for every $\base{C} \baseext\base{B}$ and every $A \in \mathrm{At}$,} \notag \\
    & && \quad\quad\text{if } \varphi \supp_{\base{C}} A \text{ and } \psi \supp_{\base{C}} A \text{, then } \supp_{\base{C}} A \tag{$\lor$} \\[1mm]
    & \supp_{\base{B}} \bot && \text{iff} \quad \supp_{\base{B}} A \text{ for every } A \in \mathrm{At} \tag{$\bot$} \\[1mm]
    & \Gamma \supp_{\base{B}} \varphi && \text{iff} \quad \text{for every $\base{C} \baseext\base{B}$, if } \supp_{\base{C}} \psi \text{ for every } \psi \in \Gamma \text{, then } \supp_{\base{C}} \varphi \tag{Inf}
\end{align*}
\vspace{2mm}
\hrule
\caption{Support in a base.}
\label{fig:support}
\end{figure}

\begin{remark}
\label{rem:foundational}
The clauses of
Figure~\ref{fig:support} contain phrases of the form `for every $A
\in \mathrm{At}$' --- in particular, the clause for falsum reads
\[
\supp_{\base{B}} \bot \qquad \text{iff} \qquad \supp_{\base{B}} A
\text{ for every } A \in \mathrm{At}.
\]
What does this mean when $\mathrm{At}$ is infinite only potentially,
not actually? Read set-theoretically, the phrase ranges over a
completed totality of atoms. Read constructively, it claims a method
which, given any atom, establishes that the atom is supported; and
this reading is available, since nothing in the clauses requires the
set-theoretic one. But the clause does not say what the method is.

Read in the spirit of Section~\ref{sec:atomic-systems}, it does: `for
every $A$' is understood as `for an arbitrary $A$', an arbitrary atom
being represented by an eigenvariable, and the clause for falsum
becomes
\[
\supp_{\base{B}} \bot \qquad \text{iff} \qquad \base{B} \vdash_O A
\quad \text{for some atom $A$ fresh for $\base{B}$,}
\]
where freshness means that $A$ does not occur in $\base{B}$. The two
readings agree on which formulae are supported; they differ in what a
verification consists in. On the first, it is whatever establishes
the universal claim; on the second, it is a single derivation
together with a renaming argument (Lemma~\ref{lem:rename}).
Section~\ref{sec:snc} shows that the second reading is available for
every clause, not only for falsum.
\end{remark}

Two structural properties of support are used throughout. The first is
immediate from the definition. The second is nearly so. From left to
right, it is clause (Inf) at the base $\Bas \cup S$, together with
clause (At). From right to left, it requires in addition that $\vdash_O$
be closed under cut on atoms, as recorded at the end of
Section~\ref{sec:atomic-systems}: if $P \vdash_O B$ for every $B \in S$
and $P, S \vdash_O A$, then $P \vdash_O A$.

\begin{proposition}[Monotonicity~\cite{Sandqvist2015IL}]
\label{prop:mono}
If $\suppB{\Bas}\varphi$ and $\BasC \baseext\Bas$, then
$\suppB{\BasC}\varphi$.
\end{proposition}

\begin{proposition}[Atomic adequacy~\cite{Sandqvist2015IL}]
\label{prop:cons}
For every atom $A$ and every finite set of atoms $S$,
\[
S \suppB{\Bas} A \qquad \mbox{iff} \qquad \Bas \cup S \vdash_O A.
\]
\end{proposition}

A third property is not needed for the proof of the main theorem, but
bears on its interpretation in Section~\ref{sec:snc}. It says that
every formula is characterised by what it licenses one to conclude, and
it generalises a reformulation of the clause for conjunction observed by
Gheorghiu et al.~\cite{IMLL}. Two consequences of clause (Inf) are used
in its proof and thereafter without comment: if $\Gamma \supp_{\BasC}
\varphi$ and $\base{D} \baseext \BasC$, then $\Gamma \supp_{\base{D}}
\varphi$; and if $\Gamma \supp_{\BasC} \varphi$ and $\supp_{\BasC} \psi$
for every $\psi \in \Gamma$, then $\supp_{\BasC} \varphi$, by taking the
extension in (Inf) to be $\BasC$ itself.

\begin{lemma}[Elimination form]
\label{lem:elim}
For every base $\Bas$ and formula $\vartheta$,
\[
\supp_{\Bas} \vartheta
\qquad \text{iff} \qquad
\text{for every } \BasC \baseext \Bas \text{ and every atom } A,\
\text{if } \vartheta \supp_{\BasC} A \text{ then } \supp_{\BasC} A.
\]
\end{lemma}

\begin{proof}
For the left-to-right direction, let $\BasC \baseext \Bas$ and $A$ be
given with $\vartheta \supp_{\BasC} A$. Proposition~\ref{prop:mono}
gives $\supp_{\BasC} \vartheta$, whence $\supp_{\BasC} A$.

The right-to-left direction is by induction on $\vartheta$, the
statement being taken for all bases at once. Suppose that, for every
$\BasC \baseext \Bas$ and atom $A$, $\vartheta \supp_{\BasC} A$ implies
$\supp_{\BasC} A$.

\medskip
\noindent\textbf{Case $\vartheta = A'$ (atom).} Since $A' \supp_{\Bas}
A'$ holds by clause (Inf), the supposition with $\BasC = \Bas$ and $A = A'$
gives $\supp_{\Bas} A'$.

\medskip
\noindent\textbf{Case $\vartheta = \bot$.} For any atom $A$, $\bot
\supp_{\Bas} A$ holds by clause ($\bot$), so the supposition with $\BasC
= \Bas$ gives $\supp_{\Bas} A$. Since $A$ was arbitrary, clause
($\bot$) gives $\supp_{\Bas} \bot$.

\medskip
\noindent\textbf{Case $\vartheta = \varphi \land \psi$.} We show
$\supp_{\Bas} \varphi$ by the induction hypothesis for $\varphi$. Let
$\BasC \baseext \Bas$ and $A$ be given with $\varphi \supp_{\BasC} A$.
Then $\varphi \land \psi \supp_{\BasC} A$: if $\base{D} \baseext \BasC$
and $\supp_{\base{D}} \varphi \land \psi$, clause ($\land$) gives
$\supp_{\base{D}} \varphi$, and $\varphi \supp_{\base{D}} A$ gives
$\supp_{\base{D}} A$. The supposition now gives $\supp_{\BasC} A$. Hence
$\supp_{\Bas} \varphi$; likewise $\supp_{\Bas} \psi$; and clause
($\land$) concludes.

\medskip
\noindent\textbf{Case $\vartheta = \varphi \lor \psi$.} Let $\BasC
\baseext \Bas$ and $A$ be given with $\varphi \supp_{\BasC} A$ and $\psi
\supp_{\BasC} A$; by clause ($\lor$) it suffices to show $\supp_{\BasC}
A$. Now $\varphi \lor \psi \supp_{\BasC} A$: if $\base{D} \baseext
\BasC$ and $\supp_{\base{D}} \varphi \lor \psi$, then $\varphi
\supp_{\base{D}} A$ and $\psi \supp_{\base{D}} A$, and clause ($\lor$)
applied at $\base{D}$ gives $\supp_{\base{D}} A$. The supposition gives
$\supp_{\BasC} A$.

\medskip
\noindent\textbf{Case $\vartheta = \varphi \to \psi$.} By clause
($\to$) it suffices to show $\varphi \supp_{\Bas} \psi$. Let $\BasC
\baseext \Bas$ with $\supp_{\BasC} \varphi$; we show $\supp_{\BasC}
\psi$ by the induction hypothesis for $\psi$ at the base $\BasC$. Let
$\base{D} \baseext \BasC$ and $A$ be given with $\psi \supp_{\base{D}}
A$. Then $\varphi \to \psi \supp_{\base{D}} A$: if $\base{E} \baseext
\base{D}$ and $\supp_{\base{E}} \varphi \to \psi$,
Proposition~\ref{prop:mono} gives $\supp_{\base{E}} \varphi$, clause
($\to$) gives $\supp_{\base{E}} \psi$, and $\psi \supp_{\base{E}} A$
gives $\supp_{\base{E}} A$. Since $\base{D} \baseext \Bas$, the
supposition gives $\supp_{\base{D}} A$.
\end{proof}

\begin{corollary}
\label{cor:elim}
For every base $\Bas$ and formulae $\varphi$, $\psi$:
\begin{enumerate}
\item[(i)] $\supp_{\Bas} \varphi \land \psi$ iff for every $\BasC
\baseext \Bas$ and atom $A$, if $\varphi, \psi \supp_{\BasC} A$ then
$\supp_{\BasC} A$;
\item[(ii)] $\supp_{\Bas} \varphi \to \psi$ iff for every $\BasC
\baseext \Bas$ and atom $A$, if $\supp_{\BasC} \varphi$ and $\psi
\supp_{\BasC} A$ then $\supp_{\BasC} A$.
\end{enumerate}
\end{corollary}

\begin{proof}
(i) Left to right: given $\BasC$ and $A$ with $\varphi, \psi
\supp_{\BasC} A$, Proposition~\ref{prop:mono} and clause ($\land$) give
$\supp_{\BasC} \varphi$ and $\supp_{\BasC} \psi$, whence $\supp_{\BasC}
A$. Right to left: by Lemma~\ref{lem:elim} it suffices to show that
$\varphi \land \psi \supp_{\BasC} A$ implies $\supp_{\BasC} A$ for every
$\BasC \baseext \Bas$ and atom $A$. If $\varphi \land \psi \supp_{\BasC}
A$, then $\varphi, \psi \supp_{\BasC} A$, since $\supp_{\base{D}}
\varphi$ and $\supp_{\base{D}} \psi$ give $\supp_{\base{D}} \varphi
\land \psi$ by clause ($\land$); the hypothesis then gives
$\supp_{\BasC} A$.

(ii) Left to right: given $\BasC$ and $A$ with $\supp_{\BasC} \varphi$
and $\psi \supp_{\BasC} A$, Proposition~\ref{prop:mono} gives
$\supp_{\BasC} \varphi \to \psi$, clause ($\to$) gives $\supp_{\BasC}
\psi$, and $\psi \supp_{\BasC} A$ gives $\supp_{\BasC} A$. Right to
left: let $\BasC \baseext \Bas$ with $\supp_{\BasC} \varphi$; by clause
($\to$) it suffices to show $\supp_{\BasC} \psi$, and by
Lemma~\ref{lem:elim} it suffices that $\psi \supp_{\base{D}} A$ implies
$\supp_{\base{D}} A$ for every $\base{D} \baseext \BasC$ and atom $A$.
Proposition~\ref{prop:mono} gives $\supp_{\base{D}} \varphi$; given
$\psi \supp_{\base{D}} A$, the hypothesis at $\base{D} \baseext \Bas$
then gives $\supp_{\base{D}} A$.
\end{proof}

Part (i) is the reformulation of Gheorghiu et al.~\cite{IMLL}; part
(ii) is the corresponding fact for implication. Together they show that
every connective, and not only $\bot$ and $\lor$, admits a clause in
elimination form.

The principal result of the framework is Sandqvist's theorem.

\begin{theorem}[Sandqvist~\cite{Sandqvist2015IL}]
\label{thm:Sandqvist}
For every formula $\varphi$,
\[
\supp_{\base{B}} \varphi \mbox{ for every base
$\base{B}$} \qquad \mbox{iff} \qquad 
\text{$\varphi$ is provable in intuitionistic propositional logic.}
\]
\end{theorem}

Soundness --- the direction from right to left --- is routine: each rule of natural deduction is shown
to be respected by support, and the result follows by induction on
derivations. Completeness --- from left to right --- is the substantive direction, and its proof
is illuminating for what follows.

Sandqvist's argument proceeds by \emph{flattening}. Fix a formula
$\varphi$ and let $\Xi$ be its set of subformulae. Construct an
injection $(-)^\flat \colon \Xi \to \mathrm{At}$ that is the identity
on atoms ($A^\flat = A$ for $A \in \Xi \cap \mathrm{At}$) and that
associates a fresh atom to each complex subformula. From this, one
constructs a base $\base{N}$ with two key properties:
\begin{itemize}
    \item[$(\dagger)$] $\chi^\flat$ and $\chi$ are semantically
    interchangeable in any extension $\base{N}' \baseext \base{N}$:
    \[
    \supp_{\base{N}'} \chi \qquad \text{iff} \qquad \base{N}' \vdash_{O} \chi^\flat;
    \]
    \item[$(\dagger\dagger)$] derivability in $\base{N}$ simulates
    provability in intuitionistic propositional logic, written $\vdash$:
    \[
    \base{N} \vdash_O \varphi^\flat \qquad \text{iff} \qquad \vdash \varphi.
    \]
\end{itemize}
Global completeness follows immediately:
\begin{align*}
\supp_{\base{B}} \varphi \text{ for all $\base{B}$}
\quad &\Longrightarrow \quad \supp_{\base{N}} \varphi \\
&\Longrightarrow  \quad \base{N} \vdash_O \varphi^\flat
\tag{$\dagger$} \\
&\Longrightarrow \quad \vdash \varphi. \tag{$\dagger\dagger$}
\end{align*}
The argument is elementary, constructive, and versatile,
having since been deployed to establish completeness for a range of
logics~\cite{FOL,SOL,IMLL,BI}.

Sandqvist's theorem is global. It leaves untouched the local question
of the introduction: what $\supp_{\base{B}} \varphi$ amounts to in a
fixed base.

The question has been partly addressed. Gheorghiu and Pym~\cite{NAF}
have generalised the flattening construction by relativising the
operator $(-)^\flat$ to a given base $\base{B}$: the fresh atom
$\chi^\flat$ assigned to each complex subformula $\chi$ of $\varphi$ is
required not to occur in $\base{B}$. The resulting simulation lemma is exactly local:
\[
\supp_{\base{B}} \varphi \qquad \text{iff} \qquad
\base{N} \cup \base{B} \vdash_O \varphi^\flat.
\]
The equivalence is a genuine local correspondence theorem, but it is
purchased at a cost. The flattening operator absorbs the entire
logical structure of $\varphi$ into the base $\base{N}$, replacing the
formula with a single fresh atom $\varphi^\flat$. Whatever relation the
connectives of $\varphi$ bear to the proof-search procedure is
invisible in the theorem itself: the procedure sees only an atom, and
the structure of $\varphi$ has been compiled away.

What is wanted, instead, is a correspondence in which the connectives
of $\varphi$ remain visible --- in which the universal quantifiers of
the support clauses correspond to identifiable steps of proof-search,
and the procedure is driven by the logical structure of the formula
rather than by a precomputed translation of it. The next section supplies it. The encoding is direct:
each connective of $\varphi$ is matched with a fixed pattern of the
clausal language, and the universal quantifiers in the support
clauses turn out to be the eigenvariables of proof-search in disguise.

\section{The Encoding and the Main Theorem}
\label{sec:snc}

The clauses of the support relation,
read as they stand, define a procedure: given a formula and a base,
decompose the formula according to its outermost connective, generate
the clauses for each connective, and consult the base when an atomic
case is reached. The procedure is, in everything but name, the
operational semantics of a logic programming language. This section makes the identification exact. It begins with the shape
of the support clauses.

The clause for disjunction is the following:
\[
\begin{array}{ll}
\supp_{\base{B}} \varphi \lor \psi \quad \text{iff} & \text{for every } \base{C} \baseext \base{B}\ \text{and every atom } A, \\[2pt]
& \text{if } \varphi \supp_{\base{C}} A\ \text{and}\ \psi \supp_{\base{C}} A,\ \text{then } \supp_{\base{C}} A.
\end{array}
\]
The clause says what one may \emph{do} with a disjunction: conclude any
atom $A$ that can be concluded from each disjunct separately. The atom
$A$ is an arbitrary conclusion which the disjunction is to deliver, and
the disjuncts function as \emph{continuations} --- branches of a
case analysis to which control is passed when $A$ is sought. In this
sense the clause is a continuation-passing definition. The reading has
independent support: Gu et al.~\cite{gu2026continuations} interpret the
minor premisses of disjunction elimination as continuations, the
conclusion being the result of the computation, and analyse Sandqvist's
completeness proof as an implicit CPS transformation.

The description calls for a word of justification, since
continuation-passing style (CPS) is in the first instance a
transformation on programs~\cite{plotkin1975}, and the semantics
transforms nothing. What is intended is the type discipline that the
transformation induces. Under CPS conversion a value of type $\alpha$
becomes a function that accepts a continuation of type $\alpha \to \rho$
and returns an answer of type $\rho$, so that the type $\alpha$ itself
becomes $(\alpha \to \rho) \to \rho$; abstracting the answer type gives
$\forall \rho\,((\alpha \to \rho) \to \rho)$, and the image of a sum type
$\alpha + \beta$ is $\forall \rho\,((\alpha \to \rho) \to (\beta \to
\rho) \to \rho)$. Griffin~\cite{griffin1990} observed that, read at the
level of types, the conversion is the double-negation translation.

The
clause ($\lor$) is this type with the answer type restricted to atoms
and the arrows read as consequence in a base. It is in this sense, and
not in the sense of a transformation applied to source programs, that we
call the clauses continuation-passing. The resulting encoding of the
connectives by implication and a propositional quantifier is older
than either reading: it is the Russell--Prawitz
translation~\cite[\S 19]{russell1903},~\cite{Prawitz2006natural}, and we
use that name when the encoding is in question as a syntactic object
rather than as a reading of the semantics.

The other connectives admit the same reading. Falsum is the trivial
case: it passes \emph{any} continuation, since $\supp_{\base{B}} \bot$
licenses the conclusion of every atom. Sandqvist's clause for
conjunction is direct, but Corollary~\ref{cor:elim}(i) shows it to be
equivalent to an elimination form of the same shape as ($\lor$), an
observation first made by Gheorghiu et al.~\cite{IMLL}; and
Corollary~\ref{cor:elim}(ii) shows that implication, too, is
characterised by what it licenses one to conclude: from $\varphi \to
\psi$ together with $\varphi$, whatever follows from $\psi$. Every
connective is thus given by its elimination behaviour. The encoding
transcribes that behaviour where Sandqvist's clauses state it, for
$\bot$ and $\lor$; for $\land$ and $\to$, whose clauses are direct, the
encoding is direct also.

CPS-shaped definitions translate directly into clauses of a
hereditary Harrop language. The case analysis on continuations becomes
a conjunction of implications; the universal quantification over
atoms becomes the quantifier over atoms; the procedure of trying the
continuation against each disjunct becomes backchaining through the
implications. We make this translation explicit.

\begin{definition}[Encoding]
\label{def:encoding}
To each formula $\varphi$ of intuitionistic propositional logic, assign a formula $\enc{\varphi}$
of the clausal language by induction on the structure of $\varphi$:
\[
  \begin{aligned}
    \enc{A} &\;=\; A \\[2pt]
    \enc{\bot} &\;=\; \forall X\, X\\[2pt]
    \enc{\varphi \to \psi} &\;=\; \enc{\varphi} \supset \enc{\psi} \\[2pt]
    \enc{\varphi \land \psi} &\;=\; \enc{\varphi} \cap \enc{\psi} \\[2pt]
    \enc{\varphi \lor \psi} &\;=\; \forall X\,\big((\enc{\varphi} \supset X) \cap (\enc{\psi} \supset X) \supset X \big)
  \end{aligned}
\]
where $X$ is a variable not occurring in $\enc{\varphi}$
or $\enc{\psi}$.
\end{definition}

The encoding follows Sandqvist's clauses one by one. Conjunction and
implication, whose clauses are direct, are encoded by the corresponding
connectives of the clausal language; falsum and disjunction, whose
clauses are elimination conditions quantifying over atoms, are encoded
by quantified formulae of the same shape. For $\bot$, $\to$, and $\lor$
the result is the Russell--Prawitz translation. Conjunction, which that
translation encodes in elimination form, is here taken directly, the
clausal language having it as a primitive; for this reason we speak of
a variant of that translation.

Nothing turns on the choice of the direct forms.
Corollary~\ref{cor:elim} shows that Sandqvist's clauses for conjunction
and implication are equivalent to their elimination forms, and the
theorem holds equally for the corresponding encodings, $\forall
X\,((\enc{\varphi} \cap \enc{\psi} \supset X) \supset X)$ and $\forall
X\,(\enc{\varphi} \cap (\enc{\psi} \supset X) \supset X)$, by arguments
modelled on the case of disjunction below. We do not need these
alternatives, and omit the verification. The direct forms are chosen
because they are the natural ones, and because they make the
corresponding cases of the proof short.

The direct encoding of implication is the reason for which the clausal
language admits compound clause heads. Whenever an implication occurs
to the left of another, its encoding $\enc{\varphi} \supset \enc{\psi}$
must be installed in a base by \textsc{Augment}, and its head
$\enc{\psi}$ is in general compound. It is the encoding adopted in the
short version of this paper~\cite{gheorghiu2026supportsearch}. No
substitution of compound formulae for atoms is involved in its
treatment --- a substitution which, as Remark~\ref{rem:nosubst} shows,
is in any case unavailable.

The universal
quantifier $\forall X$ in the clauses for $\bot$ and $\lor$
reflects, at the syntactic level, the universal quantification over
atoms in the corresponding semantic clauses. Under the operational
semantics of Section~\ref{sec:atomic-systems}, $X$ is instantiated only
with atoms, and the instantiation is governed by the eigenvariable
discipline of \textsc{Generic}: a single fresh witness suffices. The
infinity that the set-theoretic reading of the support clauses
appeared to require has become a single name-introduction step,
standard in the proof theory of quantification.

The main theorem of the paper is that the encoding is faithful: to be
supported in a base is for the encoded goal to succeed against that
base as a program.

\begin{theorem}[Support is Search]
\label{thm:snc}
For every base $\Bas$ and every formula $\varphi$,
\[
  \suppB{\Bas}\varphi \qquad\text{iff}\qquad \Bas \vdash_O \enc{\varphi}.
\]
\end{theorem}

The left-hand side is defined by quantification over every extension of
the base and every atom; the right-hand side, by six rules of search.
The theorem says that they are one relation.

Sandqvist's example from the introduction shows the theorem at work. As
a program, the base $\base{R}$ is the single program clause
\[
\forall X\,((A_1 \supset X) \cap (A_2 \supset X) \supset X),
\]
and this is the goal $\enc{A_1 \lor A_2}$ itself. The search for it
succeeds at once: \textsc{Generic} introduces a fresh atom $A$,
\textsc{Augment} adds $(A_1 \supset A) \cap (A_2 \supset A)$ to the
program, and \textsc{Backchain} on the instance of the clause at $A$
leaves, by \textsc{And}, the goals $A_1 \supset A$ and $A_2 \supset A$, each of which is
closed by \textsc{Augment}, \textsc{Backchain}, and \textsc{Fact}.

The
search for $\enc{A_1} = A_1$ does not succeed. The only resolved clause
with head $A_1$ is the instance of the clause at $A_1$, and
backchaining on it leaves, among others, the goal $A_2 \supset A_1$:
that is, $A_1$ again, from a program in which that instance is still
the only clause with head $A_1$. Every attempt regenerates its own
goal, and there is no derivation. The disjunction is supported and
neither disjunct is, and the search shows why.

The proof of the theorem requires four lemmas, to which we now turn.

\subsection*{Neutral formulae}

The encodings produced by Definition~\ref{def:encoding} have a
distinctive syntactic feature: the goal-only connective $\cup$ does not
occur in them. They are
therefore simultaneously definite formulae \emph{and} goal formulae ---
a property we name and exploit.

\begin{definition}[Neutral formulae]
\label{def:neutral}
A formula of the clausal language is \emph{neutral} if it is generated
by the grammar
\[
  N \;::=\; A \;\mid\; N \supset N \;\mid\; N \cap N \;\mid\; \forall X\, N,
\]
where $A$ ranges over atomic formulae.
\end{definition}

\begin{lemma}[Neutrality]
\label{lem:neutral} The following statements hold:
\begin{enumerate}
\item[(i)] Every neutral formula is both a definite formula and a goal
formula.
\item[(ii)] For every formula $\varphi$, the encoding
$\enc{\varphi}$ is neutral.
\end{enumerate}
\end{lemma}

\begin{proof}
(i) By induction on the grammar. An atomic formula is both a definite
formula and a goal. If $N$ and $N'$ are both, then $N \supset N'$ is a
definite formula, being of the form $G \supset D$ with $G = N$ and $D =
N'$, and a goal, being of the form $D \supset G$ with $D = N$ and $G =
N'$. Conjunction and
universal quantification are available to both classes.

(ii) By induction on $\varphi$. The cases of atoms and $\bot$ are
immediate. If $\enc{\varphi}$ and $\enc{\psi}$ are neutral, then so are
$\enc{\varphi} \supset \enc{\psi}$, $\enc{\varphi} \cap \enc{\psi}$,
$\enc{\varphi} \supset X$, and $\enc{\psi} \supset X$; hence so is
$(\enc{\varphi} \supset X) \cap (\enc{\psi} \supset X) \supset X$; and
universal quantification preserves neutrality.
\end{proof}

In particular, a base extended by the encoding of a formula is again a
base, a fact used repeatedly below. The neutral formulae are exactly
the formulae of the clausal language in which $\cup$ does not occur.
The elaboration of a neutral formula consists of neutral formulae, each
condition on $[P]$ preserving neutrality.

Before proving the encoding faithful, we ask how neutral formulae
sit within the language. Atoms are the paradigm of formulae
that are at once goals and program clauses, and neutral formulae extend that
class; one might expect them to behave like atoms in every respect, and
in particular to be substitutable for them. They are not, and the reason
is instructive.

\begin{remark}[Failure of compound instantiation]
\label{rem:nosubst}
Consider the principle: if $P \vdash_O G$ and the atom $A$ does not
occur in $P$, then $P \vdash_O G[F/A]$ for every neutral formula $F$.
Lemma~\ref{lem:rename} establishes this for atomic $F$; for compound
$F$ it fails. Let $P = \{\forall Y\,(Y \supset q)\}$ and $G = p \supset
q$, where $p$, $q$, and $a$ are distinct atoms. The query $P \vdash_O p
\supset q$ succeeds: \textsc{Augment} reduces it to $P, p \vdash_O q$;
the instance $p \supset q$ lies in $[P]$, and \textsc{Backchain} on it
leaves $P, p \vdash_O p$, which \textsc{Fact} closes.

The instance at
$F = a \supset a$ fails: \textsc{Augment} reduces $P \vdash_O (a \supset
a) \supset q$ to $P, a \supset a \vdash_O q$, and the clauses of $[P
\cup \{a \supset a\}]$ with head $q$ are exactly the formulae $B \supset
q$ with $B$ an atom, so some atom must be derivable; the only heads
available are $q$, which requires an atom in turn, and $a$, which
requires $a$. Hence $P \not\vdash_O (a \supset a) \supset q$, although,
by \textsc{Generic}, $P \vdash_O \forall X\,(X \supset q)$.

The universal goal is thus derivable while one of its compound
instances is not: the quantifier is genuinely atomic. The clausal
language is, in this respect, the operational counterpart of the system
$\mathbf{F}_{\mathrm{at}}$ of Ferreira and
Ferreira~\cite{Ferreria2013}, in which second-order quantification is
likewise restricted to atomic instantiation, and the encoding of
Definition~\ref{def:encoding} is a variant of the Russell--Prawitz translation by
which intuitionistic propositional logic is embedded in that system.
Corollary~\ref{cor:validity} below is, for the operational fragment, an
analogue of the faithfulness of that
embedding~\cite{ferreira2014faithfulness}. The methodological
consequence for the present paper is that Theorem~\ref{thm:snc} cannot
be proved by a substitution lemma. The proof below uses instead the
renaming of atoms by atoms (Lemma~\ref{lem:rename}), together with
identity and cut for neutral formulae, to which we now turn.
\end{remark}

Some of the behaviour of atoms does survive, and it is exactly the
behaviour that proof-search relies on. The rule \textsc{Fact}, which
closes an atomic goal found in the elaborated program, generalises to
neutral formulae; so does the admissibility of cut. Throughout, the
\emph{size} of a formula is the number of occurrences of atomic
formulae, connectives, and quantifiers in it; note that $F[A/X]$ has
the same size as $F$, and that $G \cap G' \supset D$ has the same size
as $G \supset (G' \supset D)$.

\begin{lemma}[Identity]
\label{lem:id}
Let $F$ be neutral. If $F \in [P]$, then $P \vdash_O F$. In particular,
$P, F \vdash_O F$ for every program $P$.
\end{lemma}

\begin{proof}
By induction on the size of $F$ and, among implications $N \supset M$
of the same size, on the size of the head $M$.

\medskip
\noindent\textbf{Case $F = A$.} Since $F$ is closed, $A$ is an atom, and
\textsc{Fact} applies.

\medskip
\noindent\textbf{Case $F = N \supset A$.} By \textsc{Augment}, it
suffices to show $P, N \vdash_O A$. Since $N \supset A \in [P] \subseteq
[P, N]$, \textsc{Backchain} reduces this to $P, N \vdash_O N$. As $N \in
[P, N]$ and $N$ is neutral of smaller size, the induction hypothesis
applies.

\medskip
\noindent\textbf{Case $F = N \supset (M_1 \cap M_2)$.} Each $N \supset
M_i$ lies in $[P]$ and is neutral of smaller size, so the induction
hypothesis gives $P \vdash_O N \supset M_i$, and
Lemma~\ref{lem:inv}(i) gives $P, N \vdash_O M_i$. \textsc{And} followed
by \textsc{Augment} concludes.

\medskip
\noindent\textbf{Case $F = N \supset (M' \supset M'')$.} The formula $N
\cap M' \supset M''$ lies in $[P]$; it is neutral, of the same size as
$F$, and its head $M''$ is smaller than that of $F$. The induction
hypothesis gives $P \vdash_O N \cap M' \supset M''$, and
Lemma~\ref{lem:inv}(i) gives $P, N \cap M' \vdash_O M''$, whence $P, N,
M' \vdash_O M''$ by the observation following
Lemma~\ref{lem:mono-op}. Two applications of \textsc{Augment} conclude.

\medskip
\noindent\textbf{Case $F = N \supset \forall X\, M$.} Take an atom $A$
fresh for $P$, $N$, and $M$. Then $N \supset M[A/X]$ lies in $[P]$ and
is neutral of smaller size, so the induction hypothesis gives $P
\vdash_O N \supset M[A/X]$, and Lemma~\ref{lem:inv}(i) gives $P, N
\vdash_O M[A/X]$. As $A$ is fresh for $P, N$ and $M$, \textsc{Generic}
gives $P, N \vdash_O \forall X\, M$, and \textsc{Augment} concludes.

\medskip
\noindent\textbf{Case $F = N_1 \cap N_2$.} Each $N_i$ lies in $[P]$ and
is neutral of smaller size, so the induction hypothesis gives $P
\vdash_O N_i$, and \textsc{And} concludes.

\medskip
\noindent\textbf{Case $F = \forall X\, N$.} Take an atom $A$ fresh for
$P$ and $N$. Then $N[A/X] \in [P]$, and $N[A/X]$ is neutral of smaller
size than $F$, so the induction hypothesis gives $P \vdash_O N[A/X]$,
and \textsc{Generic} concludes.
\end{proof}

\begin{lemma}[Elaboration]
\label{lem:elab}
Let $F$ be neutral with $P \vdash_O F$. If $F' \in [F]$, then $F'$ is
neutral, of size at most that of $F$, and $P \vdash_O F'$.
\end{lemma}

\begin{proof}
By induction on the generation of $F'$ in $[F]$, that is, in
$[\{F\}]$. If $F' = F$ there is nothing to prove. If $F' = N_i$ arises
by projection from $N_1 \cap N_2 \in [F]$, the induction hypothesis
gives $P \vdash_O N_1 \cap N_2$, and Lemma~\ref{lem:inv}(ii) gives $P
\vdash_O N_i$. If $F' = N[A/X]$ arises by instantiation from $\forall
X\, N \in [F]$, the induction hypothesis gives $P \vdash_O \forall X\,
N$, and Lemma~\ref{lem:inv}(iii) gives $P \vdash_O N[A/X]$.

The remaining three conditions concern compound heads. If $F' = N
\supset M_i$ arises from $N \supset (M_1 \cap M_2) \in [F]$, the
induction hypothesis and Lemma~\ref{lem:inv}(i) and (ii) give $P, N
\vdash_O M_i$, and \textsc{Augment} gives $P \vdash_O N \supset M_i$.
If $F' = N \cap M' \supset M''$ arises from $N \supset (M' \supset M'')
\in [F]$, the induction hypothesis and two applications of
Lemma~\ref{lem:inv}(i) give $P, N, M' \vdash_O M''$, whence $P, N \cap
M' \vdash_O M''$ by the observation following Lemma~\ref{lem:mono-op},
and \textsc{Augment} gives $P \vdash_O N \cap M' \supset M''$. If $F' =
N \supset M[A/X]$ arises from $N \supset \forall X\, M \in [F]$, the
induction hypothesis and Lemma~\ref{lem:inv}(i) and (iii) give $P, N
\vdash_O M[A/X]$, and \textsc{Augment} gives $P \vdash_O N \supset
M[A/X]$. Neutrality and the bound on size are preserved at each step.
\end{proof}

\begin{lemma}[Cut]
\label{lem:cut}
Let $F$ be neutral. If $P \vdash_O F$ and $Q, F \vdash_O G$, then $P, Q
\vdash_O G$.
\end{lemma}

\begin{proof}
By induction on the size of $F$ and, within it, on the height of the
derivation of $Q, F \vdash_O G$; the statement is taken for all $P$,
$Q$, and $G$ at once. Elaboration distributes over union, so $[Q, F] =
[Q] \cup [F]$. We proceed by cases on the last rule of the derivation of
$Q, F \vdash_O G$.

\medskip
\noindent\textbf{Case \textsc{Fact}: $G = A \in [Q] \cup [F]$.} If $A
\in [Q]$, then $A \in [P, Q]$ and \textsc{Fact} applies. If $A \in [F]$,
Lemma~\ref{lem:elab} gives $P \vdash_O A$, and Lemma~\ref{lem:mono-op}
gives $P, Q \vdash_O A$.

\medskip
\noindent\textbf{Case \textsc{Backchain}: $G = A$, with $G' \supset A
\in [Q] \cup [F]$ and $Q, F \vdash_O G'$ by a shorter derivation.} The
secondary induction hypothesis gives $P, Q \vdash_O G'$. If $G' \supset
A \in [Q]$, then $G' \supset A \in [P, Q]$ and \textsc{Backchain}
applies. If $G' \supset A \in [F]$, then by Lemma~\ref{lem:elab} $G'
\supset A$ is neutral of size at most that of $F$, so $G'$ is neutral of
smaller size than $F$, and $P \vdash_O G' \supset A$, whence
Lemma~\ref{lem:inv}(i) gives $P, G' \vdash_O A$. The primary induction
hypothesis, applied to the cut formula $G'$ with $P, Q \vdash_O G'$ and
$P, G' \vdash_O A$, gives $P, Q, P \vdash_O A$, that is, $P, Q \vdash_O
A$.

\medskip
\noindent\textbf{Cases \textsc{And} and \textsc{Or}.} Immediate from the
secondary induction hypothesis.

\medskip
\noindent\textbf{Case \textsc{Augment}: $G = D \supset G'$, with $Q, F,
D \vdash_O G'$ by a shorter derivation.} The secondary induction
hypothesis, with $Q \cup \{D\}$ in place of $Q$, gives $P, Q, D \vdash_O
G'$, and \textsc{Augment} concludes.

\medskip
\noindent\textbf{Case \textsc{Generic}: $G = \forall X\, G'$, with $Q,
F \vdash_O G'[A/X]$ by a shorter derivation, for some $A$ fresh for
$Q$, $F$, and $G'$.} Choose an atom $B$ fresh for $P$, $Q$, $F$, and
$G'$. Lemma~\ref{lem:rename} applied to the renaming $A \mapsto B$
gives $Q, F \vdash_O G'[B/X]$ by a derivation of the same height, since
$A$ occurs in none of $Q$, $F$, $G'$. The secondary induction
hypothesis gives $P, Q \vdash_O G'[B/X]$, and as $B$ is fresh for $P,
Q$ and $G'$, \textsc{Generic} concludes.
\end{proof}

Neither Lemma~\ref{lem:id} nor Lemma~\ref{lem:cut} would be well formed
without Lemma~\ref{lem:neutral}(i): in each, $F$ occurs once as a goal
and once as a clause of a program. This is where membership in both
syntactic classes earns its keep. The arguments are those of the
admissibility of cut for uniform proofs~\cite{miller1991uniform},
adapted to the atomic quantifier.

\subsection*{Proof of the main theorem}

We are now in a position to prove Theorem~\ref{thm:snc}.

\begin{proof}[Proof of Theorem~\ref{thm:snc}]
By simultaneous induction on the structure of $\varphi$, establishing
both directions together; the induction hypothesis for a subformula is
used in both directions. Monotonicity of support
(Proposition~\ref{prop:mono}) and of $\vdash_O$
(Lemma~\ref{lem:mono-op}) under base extension are used freely, as is
the fact (Lemma~\ref{lem:neutral}) that a base extended by neutral
formulae is again a base.

\medskip
\noindent\textbf{Case $\varphi = A$.} Immediate: clause (At) of
Figure~\ref{fig:support} gives $\supp_{\Bas} A$ iff $\Bas \vdash_O A$,
and $\enc{A} = A$.

\medskip
\noindent\textbf{Case $\varphi = \bot$.}
\begin{itemize}
\item[$(\Rightarrow)$] Suppose $\supp_{\Bas} \bot$. By clause
($\bot$), $\supp_{\Bas} A$ for every atom $A$, hence $\Bas \vdash_O A$
for every atom $A$. Taking $A$ fresh for $\Bas$, \textsc{Generic} gives
$\Bas \vdash_O \forall X\, X = \enc{\bot}$.
\item[$(\Leftarrow)$] Suppose $\Bas \vdash_O \forall X\, X$.
Lemma~\ref{lem:inv}(iii) gives $\Bas \vdash_O A$ for every atom $A$,
hence $\supp_{\Bas} A$ for every $A$, and clause ($\bot$) gives
$\supp_{\Bas} \bot$.
\end{itemize}

\medskip
\noindent\textbf{Case $\varphi = \varphi_1 \to \varphi_2$.}
\begin{itemize}
\item[$(\Rightarrow)$] Suppose $\supp_{\Bas} \varphi_1 \to \varphi_2$,
so that $\varphi_1 \supp_{\Bas} \varphi_2$ by clause ($\to$). Set $\BasC
\coloneqq \Bas \cup \{\enc{\varphi_1}\}$. Lemma~\ref{lem:id} gives
$\BasC \vdash_O \enc{\varphi_1}$, and the induction hypothesis gives
$\supp_{\BasC} \varphi_1$. As $\BasC \baseext \Bas$, clause (Inf) yields
$\supp_{\BasC} \varphi_2$, and the induction hypothesis returns $\BasC
\vdash_O \enc{\varphi_2}$. \textsc{Augment} gives $\Bas \vdash_O
\enc{\varphi_1} \supset \enc{\varphi_2}$.

\item[$(\Leftarrow)$] Suppose $\Bas \vdash_O \enc{\varphi_1} \supset
\enc{\varphi_2}$, so that $\Bas, \enc{\varphi_1} \vdash_O
\enc{\varphi_2}$ by Lemma~\ref{lem:inv}(i). By clauses ($\to$) and
(Inf), it suffices to show, for every $\BasC \baseext \Bas$ with
$\supp_{\BasC} \varphi_1$, that $\supp_{\BasC} \varphi_2$. The induction
hypothesis gives $\BasC \vdash_O \enc{\varphi_1}$. As $\enc{\varphi_1}$
is neutral, cut (Lemma~\ref{lem:cut}) applied to $\BasC \vdash_O
\enc{\varphi_1}$ and $\Bas, \enc{\varphi_1} \vdash_O \enc{\varphi_2}$
gives $\BasC, \Bas \vdash_O \enc{\varphi_2}$, that is, $\BasC \vdash_O
\enc{\varphi_2}$, and the induction hypothesis returns $\supp_{\BasC}
\varphi_2$.
\end{itemize}

\medskip
\noindent\textbf{Case $\varphi = \varphi_1 \land \varphi_2$.} By clause
($\land$), $\supp_{\Bas} \varphi_1 \land \varphi_2$ iff $\supp_{\Bas}
\varphi_i$ for $i = 1, 2$; by the induction hypothesis, iff $\Bas
\vdash_O \enc{\varphi_i}$ for $i = 1, 2$; and by \textsc{And} and
Lemma~\ref{lem:inv}(ii), iff $\Bas \vdash_O \enc{\varphi_1} \cap
\enc{\varphi_2} = \enc{\varphi_1 \land \varphi_2}$.

\medskip
\noindent\textbf{Case $\varphi = \varphi_1 \lor \varphi_2$.}
\begin{itemize}
\item[$(\Rightarrow)$] Suppose $\supp_{\Bas} \varphi_1 \lor
\varphi_2$. We must show
\[
  \Bas \vdash_O \forall X\,\bigl((\enc{\varphi_1} \supset X) \cap (\enc{\varphi_2} \supset X) \supset X\bigr).
\]
Take an atom $A$ fresh for $\Bas$, $\enc{\varphi_1}$, and
$\enc{\varphi_2}$; set $D \coloneqq (\enc{\varphi_1} \supset A) \cap
(\enc{\varphi_2} \supset A)$ and $\BasC \coloneqq \Bas \cup \{D\}$. By
\textsc{Generic} followed by \textsc{Augment}, it suffices to show
$\BasC \vdash_O A$. We claim $\varphi_i \supp_{\BasC} A$ for each $i$.
Indeed, take $\BasC' \baseext \BasC$ with $\supp_{\BasC'} \varphi_i$;
the induction hypothesis gives $\BasC' \vdash_O \enc{\varphi_i}$; since
$\enc{\varphi_i} \supset A \in [D] \subseteq [\BasC']$,
\textsc{Backchain} yields $\BasC' \vdash_O A$, which is $\supp_{\BasC'}
A$. Since $\BasC' \baseext \BasC$ was arbitrary, the claim holds.
Applying the supposition $\supp_{\Bas} \varphi_1 \lor \varphi_2$ at the
extension $\BasC$ and the atom $A$, clause ($\lor$) gives
$\supp_{\BasC} A$ --- that is, $\BasC \vdash_O A$, as required.

\item[$(\Leftarrow)$] Suppose $\Bas \vdash_O \enc{\varphi_1 \lor
\varphi_2}$, and let $\BasC \baseext \Bas$ and an atom $A$ be given
with $\varphi_i \supp_{\BasC} A$ for $i = 1, 2$; by clause ($\lor$) it
suffices to show $\BasC \vdash_O A$. By monotonicity, $\BasC \vdash_O
\enc{\varphi_1 \lor \varphi_2}$, and Lemma~\ref{lem:inv}(iii) and (i)
give
\[
\BasC, D \vdash_O A,
\qquad \text{where } D \coloneqq (\enc{\varphi_1} \supset A) \cap (\enc{\varphi_2} \supset A).
\]
We claim $\BasC \vdash_O \enc{\varphi_i} \supset A$ for each $i$. Let
$\BasC_i \coloneqq \BasC \cup \{\enc{\varphi_i}\}$.
Lemma~\ref{lem:id} gives $\BasC_i \vdash_O \enc{\varphi_i}$, the
induction hypothesis gives $\supp_{\BasC_i} \varphi_i$, and since
$\BasC_i \baseext \BasC$, the hypothesis $\varphi_i \supp_{\BasC} A$
gives $\supp_{\BasC_i} A$, that is, $\BasC_i \vdash_O A$;
\textsc{Augment} establishes the claim. \textsc{And} then gives $\BasC
\vdash_O D$, and as $D$ is neutral, cut (Lemma~\ref{lem:cut}) applied
to $\BasC \vdash_O D$ and $\BasC, D \vdash_O A$ gives $\BasC \vdash_O
A$.
\end{itemize}

This completes the induction.
\end{proof}

Observe that no compound instantiation of a quantifier occurs anywhere
in the proof, in keeping with Remark~\ref{rem:nosubst}: the fresh atom
introduced by \textsc{Generic} is moved to the atom supplied by the
semantic clause by Lemma~\ref{lem:inv}(iii), which is to say by renaming,
and the compound formulae required as hypotheses are supplied by identity
and discharged by cut.

Since every base extends the empty base, Proposition~\ref{prop:mono}
shows that $\varphi$ is supported in every base iff it is supported in
the empty base. Theorem~\ref{thm:Sandqvist}
therefore yields the following.

\begin{corollary}
\label{cor:validity}
A formula $\varphi$ is provable in intuitionistic propositional logic iff $\emptyset \vdash_O
\enc{\varphi}$.
\end{corollary}

\begin{remark}[Miller's Kripke model]
\label{rem:miller}
The proof of Theorem~\ref{thm:snc} follows a precedent in the proof
theory of logic programming. Miller~\cite{miller1992abstract}, and again
in~\cite[\S 5.6]{miller2025ptlp}, gives a Kripke-style semantics for
hereditary Harrop programs in which the worlds are the programs
themselves, ordered by extension, and shows that satisfaction at a
program coincides with uniform provability from it; Hodas and
Miller~\cite{hodas1994logic} extend the construction to a fragment of
intuitionistic linear logic. The skeleton of the argument is the same
as ours: monotonicity of both relations under extension of the
program, initial sequents and the admissibility of cut for a
restricted class of formulae, and a simultaneous induction on the
structure of the formula. The reader who knows Miller's argument will
find the present one familiar, and the debt is gladly acknowledged.

What has to be shown before the method applies is different. Here the
semantic side is Sandqvist's base-extension semantics for an
independently given object logic, whose connectives do not belong to
the clausal language, and whose clauses for $\lor$ and $\bot$ are stated
as elimination conditions quantifying over atoms, not as the
satisfaction clauses of a Kripke model. That they agree with the
search behaviour of the encoded formulae is shown in the corresponding cases of
the proof, and it is where the encoding of
Definition~\ref{def:encoding} and the atomic quantifier are needed.

The machinery, then, is in essentials Miller's. What is new is the
identification: Sandqvist's support relation, motivated independently
and defined without reference to proof-search, is that computational
structure.
\end{remark}

The theorem identifies support in a fixed base with proof-search in a
fixed program. The universal quantifiers in the support clauses --- the
universal over base extensions in ($\to$) and (Inf), the universals
over extensions and atoms in ($\lor$), and the universal over atoms in
($\bot$) --- have all become
steps of search: \textsc{Augment} adds to the program, and
\textsc{Generic} takes a fresh atom. None of them ranges over a
completed totality. On this reading the semantics is finitary,
mechanical, and exact.

\section{Discussion} \label{sec:discussion}

The result of Section~\ref{sec:snc} admits a one-line statement:
support in a fixed base is proof-search in a fixed program. The
remainder of this paper makes three observations about what this
statement means and what it makes possible.

The first concerns the meta-theory of base-extension semantics
itself. The clauses defining support quantify universally over base
extensions and atoms; read set-theoretically, they range over a
completed totality of inferential contexts, populated in advance and
quantified over from without. That reading is not forced. A
constructivist asserts universal claims freely, understanding them as
claims to a method, and we do not claim that Sandqvist's meta-theory
anywhere makes a constructively invalid inference. But a reading merely
available is one thing and a method exhibited another. A semantics
with anti-realist ambitions ought to supply the second.

Theorem~\ref{thm:snc} supplies it. Each universal quantifier in the
support clauses is identified with a specific operational mechanism:
the universal over base extensions in clauses ($\to$) and (Inf)
becomes the augmentation of the program in \textsc{Augment}; the
universals over atoms in ($\lor$) and ($\bot$) become the
freshness condition in \textsc{Generic}. None of these mechanisms
ranges over a completed totality. Each is a single step, standard in the
proof theory of quantification. The semantics is constructive not
merely in aspiration but in execution.

The set-theoretic and the
constructive readings of the quantifiers agree on which formulae are
supported in which bases (Remark~\ref{rem:foundational}); the theorem
does not separate them extensionally. What it fixes is the
intensional content of support: what it is to verify
$\supp_{\base{B}} \varphi$. On the first reading, it is the truth of
a universal claim about a class of inferential contexts; on the
second, it is the existence of a successful run of a finite search.
The theorem shows that base-extension semantics, read through its own
clauses, defines the second.

The second observation concerns the limits of the result, and is
perhaps more interesting than the result itself. Theorem~\ref{thm:snc}
depends essentially on the higher-level structure of atomic systems
--- that is, on the availability of hypothetical premisses in atomic
rules. Without higher-level rules, the encoding
of Section~\ref{sec:snc} cannot be carried out: hereditary Harrop
formulae collapse to Horn clauses, and the elimination-driven
clauses of the support relation lose their natural translation. The
correspondence between support and proof-search dissolves with them.

This dependence is not a technical artefact; it tracks a genuine
logical fact. Sandqvist~\cite{Sandqvist2009CL} has shown that
restricting base-extension semantics to first-level atomic systems
yields a characterisation not of intuitionistic but of classical
propositional logic. The same restriction that destroys our encoding
also destroys the intuitionistic character of the semantics. The two
phenomena are aspects of a single fact: \emph{higher-level rules are
to atomic systems what hypothetical goals are to logic programming};
and the difference between having them and lacking them is the
difference between intuitionistic and classical logic. The
present paper establishes the intuitionistic half of this
correspondence.

The classical side remains open, and the gap is significant. Dummett's
programme leaves room, in principle, for a vindication of classical
logic on anti-realist grounds, without prior commitment to bivalence,
though Dummett himself doubted that one was to be had and concluded
that classical logic should be given up.
Sandqvist~\cite{Sandqvist2005inferentialist,Sandqvist2009CL} offers his
classical semantics as just such a vindication.

If support in
the classical setting is genuinely a coherent anti-realist notion,
then it should admit a constructive, search-based reading analogous
to the one we have given for the intuitionistic case. If it does not
--- if the classical correspondence cannot be made computational
without falling back on Horn clauses and a model-theoretic reading of
the universal quantifiers --- then the anti-realist credentials of
classical base-extension semantics are not what they have been
advertised to be. We leave this question open, but flag it as the
natural sequel to the present work, and as a test case for the
foundational claims of the broader programme.

The third observation looks outward. Base-extension semantics has
been developed for a range of logics beyond intuitionistic
propositional logic --- first- and higher-order intuitionistic
logic~\cite{FOL,SOL}, intuitionistic multiplicative linear
logic~\cite{IMLL}, the logic of bunched implications~\cite{BI}, and
others --- and the pattern of the semantic clauses is the same in
every case. The meaning of each connective is given by its elimination
behaviour, expressed as a universal quantification over base
extensions. The shape that drives our encoding is not peculiar to
intuitionistic propositional logic: it is constitutive of the framework.

This suggests a research programme. Wherever the elimination-driven
shape holds, and wherever a suitable analogue of higher-level atomic
systems is available, the present result indicates that support
should reduce to proof-search in a corresponding logic programming
language. Carrying this out systematically would require, for each
logic, identifying the right class of program clauses, the right encoding of
formulae as goals, and the right notion of uniform proof. In some
cases the candidate is already to hand: for first-order logic, the
first-order hereditary Harrop programs of Miller et
al.~\cite{miller1991uniform}; for intuitionistic multiplicative
linear logic, the linear hereditary Harrop programs of Hodas and
Miller~\cite{hodas1994logic}. The payoff
would be substantial: a uniform computational foundation for
proof-theoretic semantics, in which the philosophical content of the
framework is at every point underwritten by a transparent
mathematical object.

Proof-theoretic semantics
has long been received --- with some justice, given its expository
register --- as a primarily philosophical undertaking, of interest
chiefly to those already persuaded of its anti-realist motivations.
The present result, and the programme it sketches, suggests a
different reading. Under the surface, and largely without remarking on it,
proof-theoretic semantics has been writing computational
specifications for the logical constants. On the reading we have
given, Sandqvist's clauses for intuitionistic logic are a logic
program; the theorem does not impose that structure but finds it there.
Whether the same is true of
the framework as a whole is a question we leave open, but one we now
have the means to ask precisely.

\bibliographystyle{siam}
\bibliography{references}

@article{NAF,
	title        = {{Definite Formulae, Negation-as-Failure, and the Base-extension Semantics for Intuitionistic Propositional Logic}},
	author       = {Gheorghiu, Alexander V. and Pym, David J.},
	year         = 2023,
	journal      = {{Bulletin of the Section of Logic}},
	publisher    = {Lodz University Press}
}

@article{BI,
	title        = {{Proof-theoretic Semantics for the Logic of Bunched Implications}},
	author       = {Gheorghiu, Alexander V. and Gu, Tao and Pym, David J.},
	year         = 2025,
	journal      = {Studia Logica},
	doi          = {10.1007/s11225-025-10202-z}
}

@inproceedings{IMLL,
	title        = {{Proof-theoretic Semantics for Intuitionistic Multiplicative Linear Logic}},
	author       = {Gheorghiu, Alexander V. and Gu, Tao and Pym, David J.},
	year         = 2023,
	booktitle    = {{Automated Reasoning with Analytic Tableaux and Related Methods --- TABLEAUX}},
	publisher    = {Springer},
	pages        = {367--385},
	isbn         = {978-3-031-43513-3},
	editor       = {Ramanayake, Revantha and Urban, Josef}
}

@article{FOL,
	title        = {{Proof-theoretic Semantics for First-order Logic}},
	author       = {Alexander V. Gheorghiu},
	year         = 2025,
	journal      = {{Logic Journal of the IGPL}},
	volume       = 33,
	number       = 5
}

@misc{SOL,
	title        = {{Proof-theoretic Semantics for Second-order Logic}},
	author       = {Gheorghiu, Alexander V. and Pym, David J.},
	year         = 2025,
	month        = {August},
	doi          = {10.48550/arXiv.2508.07786},
	url          = {https://arxiv.org/abs/2508.07786},
	howpublished = {arXiv preprint},
	eprint       = {2508.07786},
	archiveprefix = {arXiv},
	primaryclass = {math.LO}
}

@incollection{Prawitz1973towards,
	title        = {{Towards a Foundation of a General Proof Theory}},
	author       = {Dag Prawitz},
	year         = 1973,
	booktitle    = {{Studies in Logic and the Foundations of Mathematics}},
	publisher    = {Elsevier},
	volume       = 74,
	pages        = {225--250},
	city         = {Amsterdam}
}

@book{Prawitz2006natural,
	title        = {{Natural Deduction: A Proof-theoretical Study}},
	author       = {Dag Prawitz},
	year         = {2006 [1965]},
	publisher    = {Dover Publications},
	city         = {Mineola, New York}
}

@incollection{Prawitz1971ideas,
	title        = {{Ideas and Results in Proof Theory}},
	author       = {Dag Prawitz},
	year         = 1971,
	booktitle    = {{Studies in Logic and the Foundations of Mathematics}},
	publisher    = {Elsevier},
	volume       = 63,
	pages        = {235--307}
}

@incollection{Dummett1976,
	title        = {{What Is a Theory of Meaning?}},
	author       = {Michael Dummett},
	year         = 1996,
	booktitle    = {{The Seas of Language}},
	publisher    = {Oxford University Press},
	doi          = {10.1093/0198236212.003.0001}
}

@incollection{Mints1988,
  author    = {Mints, Grigori},
  title     = {Gentzen-type Systems and Resolution Rules {P}art {I}: {P}ropositional Logic},
  booktitle = {Computer Logic---{COLOG}-88},
  editor    = {Martin-L\"{o}f, Per and Mints, Grigori},
  series    = {Lecture Notes in Computer Science},
  publisher = {Springer},
  year      = {1988},
  pages     = {198--231},
}

@book{miller2012programming,
  author    = {Miller, Dale and Nadathur, Gopalan},
  title     = {Programming with Higher-Order Logic},
  publisher = {Cambridge University Press},
  year      = {2012},
  address   = {Cambridge},
  doi       = {10.1017/CBO9781139021326}
}

@misc{gheorghiu2026supportsearch,
      title={Support is Search}, 
      author={Alexander V. Gheorghiu},
      year={2026},
      howpublished={arXiv:2603.13018},
      note = {Accessed August 2026}
}

@book{kowalski1979logic,
  author    = {Kowalski, Robert},
  title     = {Logic for Problem Solving},
  publisher = {North-Holland},
  address   = {New York},
  series    = {Artificial Intelligence Series},
  number    = {7},
  year      = {1979}
}

@article{hodas1994logic,
  author  = {Hodas, Joshua S. and Miller, Dale},
  title   = {Logic Programming in a Fragment of Intuitionistic Linear Logic},
  journal = {Information and Computation},
  volume  = {110},
  number  = {2},
  pages   = {327--365},
  year    = {1994},
  doi     = {10.1006/inco.1994.1036}
}

@inproceedings{griffin1990,
  author    = {Griffin, Timothy G.},
  title     = {A Formulae-as-Types Notion of Control},
  booktitle = {Proceedings of the 17th ACM SIGPLAN-SIGACT Symposium on
               Principles of Programming Languages (POPL '90)},
  pages     = {47--58},
  publisher = {ACM},
  year      = {1990},
  doi       = {10.1145/96709.96714}
}

@book{russell1903,
  author    = {Russell, Bertrand},
  title     = {The Principles of Mathematics},
  publisher = {Cambridge University Press},
  address   = {Cambridge},
  year      = {1903}
}

@incollection{Dummett1978,
	title        = {{The Justification of Deduction}},
	author       = {Michael Dummett},
	year         = 1978,
	booktitle    = {{Truth and other Enigmas}},
	publisher    = {Duckworth \& Co},
	pages        = 318
}

@book{Dummett1991logical,
	title        = {{The Logical Basis of Metaphysics}},
	author       = {Michael Dummett},
	year         = 1991,
	publisher    = {Harvard University Press}
}

@article{Ferreria2013,
	title        = {{Atomic Polymorphism}},
	author       = {Fernando Ferreira and Gilda Ferreira},
	year         = 2013,
	journal      = {{The Journal of Symbolic Logic}},
	publisher    = {Cambridge University Press},
	volume       = 78,
	pages        = {260--274},
	issn         = {00224812, 19435886},
	url          = {http://www.jstor.org/stable/43303648}
}

@article{ferreira2014faithfulness,
	title        = {{The faithfulness of atomic polymorphism}},
	author       = {Ferreira, Fernando and Ferreira, Gilda},
	year         = 2014,
	journal      = {{Trends in Logic XIII}},
	publisher    = {Lodz University Press}
}

@article{hallnas1991proof,
	title        = {{A Proof-theoretic Approach to Logic Programming: II. Programs as Definitions}},
	author       = {Halln{\"a}s, Lars and Schroeder-Heister, Peter},
	year         = 1991,
	journal      = {{Journal of Logic and Computation}},
	publisher    = {Oxford University Press},
	volume       = 1,
	number       = 5,
	pages        = {635--660}
}

@article{hallnas1990proof,
	title        = {{A Proof-theoretic Approach to Logic Programming: I. Clauses as Rules}},
	author       = {Halln{\"a}s, Lars and Schroeder-Heister, Peter},
	year         = 1990,
	journal      = {{Journal of Logic and Computation}},
	publisher    = {Oxford University Press},
	volume       = 1,
	number       = 2,
	pages        = {261--283}
}

@article{harrop1960concerning,
	title        = {{ { Concerning formulas of the types $A \to B \lor C$, $A\to(Ex) B (x)$ in intuitionistic formal systems } }},
	author       = {Harrop, Ronald},
	year         = 1960,
	journal      = {{The Journal of Symbolic Logic}},
	publisher    = {Cambridge University Press},
	volume       = 25,
	number       = 1,
	pages        = {27--32}
}

@article{miller1989logical,
	title        = {{A Logical Analysis of Modules in Logic Programming}},
	author       = {Miller, Dale},
	year         = 1989,
	journal      = {{The Journal of Logic Programming}},
	publisher    = {Elsevier},
	volume       = 6,
	number       = {1-2},
	pages        = {79--108}
}

@incollection{Schroeder2016atomic,
	title        = {{Atomic Systems in Proof-Theoretic Semantics: Two Approaches}},
	author       = {Thomas Piecha and Peter Schroeder{-}Heister},
	year         = 2016,
	booktitle    = {{Epistemology, Knowledge and the Impact of Interaction}},
	publisher    = {Springer Verlag},
	editor       = {\'{A}ngel Nepomuceno Fern\'{a}ndez and Olga Pombo Martins and Juan Redmond}
}

@article{Piecha2015failure,
	title        = {{Failure of Completeness in Proof-theoretic Semantics}},
	author       = {Piecha, Thomas and de Campos Sanz, Wagner and Schroeder-Heister, Peter},
	year         = 2015,
	journal      = {{Journal of Philosophical Logic}},
	publisher    = {Springer},
	volume       = 44,
	number       = 3,
	pages        = {321--335}
}

@incollection{Piecha2017definitional,
	title        = {{The Definitional View of Atomic Systems in Proof-theoretic Semantics}},
	author       = {Thomas Piecha and Peter Schroeder{-}Heister},
	year         = 2017,
	booktitle    = {{The Logica Yearbook 2016}},
	publisher    = {College Publications London},
	pages        = {185--200}
}

@article{Piecha2019incompleteness,
	title        = {{ { Incompleteness of Intuitionistic Propositional Logic with Respect to Proof-theoretic Semantics } }},
	author       = {Thomas Piecha and Peter Schroeder{-}Heister},
	year         = 2019,
	journal      = {{Studia Logica}},
	publisher    = {Springer},
	volume       = 107,
	number       = 1,
	pages        = {233--246}
}

@incollection{Piecha2016completeness,
	title        = {{Completeness in Proof-theoretic Semantics}},
	author       = {Piecha, Thomas},
	year         = 2016,
	booktitle    = {{Advances in Proof-theoretic Semantics}},
	publisher    = {Springer},
	pages        = {231--251},
	editors      = {Thomas Piecha and Peter Schroeder{-}Heister}
}

@article{MINTS,
	title        = {{On an Inferential Semantics for Intuitionistic Sentential Logic}},
	author       = {Alexander V. Gheorghiu},
	year         = 2025,
	journal      = {{Journal of Logic and Computation}},
	volume       = 35,
	number       = 4
}

@book{pym2004reductive,
	title        = {{Reductive Logic and Proof-search: Proof Theory, Semantics, and Control}},
	author       = {Pym, David J. and Ritter, Eike},
	year         = 2004,
	publisher    = {Oxford University Press}
}

@incollection{SEP-PtS,
	title        = {{Proof-Theoretic Semantics}},
	author       = {Schroeder-Heister, Peter},
	year         = 2018,
	booktitle    = {{The Stanford Encyclopedia of Philosophy}},
	publisher    = {Metaphysics Research Lab, Stanford University},
	editor       = {Edward N. Zalta},
	howpublished = {\url{https://plato.stanford.edu/archives/spr2018/entries/proof-theoretic-semantics/}},
	edition      = {{S}pring 2018}
}

@article{Schroeder2006validity,
	title        = {{Validity Concepts in Proof-theoretic Semantics}},
	author       = {Schroeder-Heister, Peter},
	year         = 2006,
	journal      = {{Synthese}},
	publisher    = {Springer},
	volume       = 148,
	number       = 3,
	pages        = {525--571}
}

@incollection{Schroeder2007modelvsproof,
	title        = {{Proof-Theoretic versus Model-Theoretic Consequence}},
	author       = {Schroeder-Heister, Peter},
	year         = 2008,
	booktitle    = {{The Logica Yearbook 2007}},
	publisher    = {Filosofia},
	editor       = {Michal Pelis}
}

@article{Sandqvist2015IL,
	title        = {{Base-extension Semantics for Intuitionistic Sentential Logic}},
	author       = {Tor Sandqvist},
	year         = 2015,
	journal      = {{Logic Journal of the IGPL}},
	volume       = 23,
	number       = 5,
	pages        = {719--731},
	doi          = {10.1093/jigpal/jzv021}
}

@article{Sandqvist2009CL,
	title        = {{Classical Logic without Bivalence}},
	author       = {Tor Sandqvist},
	year         = 2009,
	journal      = {{Analysis}},
	publisher    = {Oxford University Press},
	volume       = 69,
	number       = 2,
	pages        = {211--218},
	urldate      = {2022-05-03}
}

@phdthesis{Sandqvist2005inferentialist,
	title        = {{An Inferentialist Interpretation of Classical Logic}},
	author       = {Tor Sandqvist},
	year         = 2005,
	school       = {Uppsala University}
}

@article{Stafford2021,
	title        = {{Proof-Theoretic Semantics and Inquisitive Logic}},
	author       = {Will Stafford},
	year         = 2021,
	journal      = {{Journal of Philosophical Logic}},
	pages        = {}
}

@article{miller1991uniform,
	title        = {{Uniform proofs as a foundation for logic programming}},
	author       = {Miller, Dale and Nadathur, Gopalan and Pfenning, Frank and Scedrov, Andre},
	year         = 1991,
	journal      = {{Annals of Pure and Applied Logic}},
	publisher    = {Elsevier},
	volume       = 51,
	number       = {1-2},
	pages        = {125--157}
}

@article{plotkin1975,
  author  = {Plotkin, Gordon D.},
  title   = {Call-by-name, Call-by-value and the $\lambda$-calculus},
  journal = {Theoretical Computer Science},
  volume  = {1},
  number  = {2},
  pages   = {125--159},
  year    = {1975}
}

@incollection{miller1992abstract,
  author    = {Miller, Dale},
  title     = {Abstract Syntax and Logic Programming},
  booktitle = {Logic Programming: Proceedings of the First and Second Russian Conferences on Logic Programming},
  series    = {Lecture Notes in Artificial Intelligence},
  volume    = {592},
  pages     = {322--337},
  publisher = {Springer},
  year      = {1992}
}

@book{miller2025ptlp,
  author    = {Miller, Dale},
  title     = {Proof Theory and Logic Programming: Computation as Proof Search},
  publisher = {Cambridge University Press},
  address   = {Cambridge},
  year      = {2025}
}

@misc{gu2026continuations,
  author       = {Gu, Tao and Pym, David J. and Ritter, Eike and Robinson, Edmund},
  title        = {{Continuations and Completeness in Proof-theoretic Semantics}},
  year         = {2026},
  note         = {arXiv:2605.04689},
}

\end{document}